\documentclass{birkau}
\usepackage{amsmath}
\usepackage{amssymb,url}
\usepackage{amsthm}
\usepackage{latexsym}
\usepackage{epsfig}
\usepackage{latexsym}
\usepackage{verbatim}
\usepackage{enumitem}
\numberwithin{equation}{section}
\begin{document}
%%%%%%%%%%%%%%%%%%%%%%%%%%%%%%%%%
% Enviroments
\theoremstyle{plain}
\newtheorem{theorem}{Theorem}[section]
\newtheorem{corollary}[theorem]{Corollary}
\newtheorem{prop}[theorem]{Proposition}
\newtheorem{problem}[theorem]{Problem}
\newtheorem{lemma}[theorem]{Lemma}
\newtheorem{remark}[theorem]{Remark}
\newtheorem{observation}[theorem]{Observation}
\newtheorem{defin}{Definition}
\newtheorem{example}[theorem]{Example}
\newtheorem{conj}{Conjecture}
\newcommand{\PR}{\noindent {\bf Proof:\ }} % beginning of a
\def\EPR{\hfill $\Box$\linebreak\vskip.5mm} % end of a proof
%%%%%%%%%%%%%%%%%%%%%%%%%%%%%%%%%%%
% Operators
\def\Pol{{\sf Pol}}
\def\mPol{{\sf MPol}}
\def\Polo{{\sf Pol}_1}
\def\PPol{{\sf pPol\;}}
\def\Inv{{\sf Inv}}
\def\mInv{{\sf MInv}}
\def\Clo{{\sf Clo}\;}
\def\Con{{\sf Con}}
\def\concom{{\sf Concom}\;}
\def\End{{\sf End}\;}
\def\Sub{{\sf Sub}}
\def\Im{{\sf Im}}
\def\Ker{{\sf Ker}\;}
\def\H{{\sf H}}
\def\S{{\sf S}}
\def\D{{\sf P}}
\def\I{{\sf I}}
\def\Var{{\sf var}}
\def\PVar{{\sf pvar}}
\def\fin#1{{#1}_{\rm fin}}
\def\P{{\sf P}}
\def\Pfin{{\sf P_{\rm fin}}}
\def\Id{{\sf Id}}
\def\R{{\rm R}}
\def\F{{\rm F}}
\def\Term{{\sf Term}}
\def\var#1{{\sf var}(#1)}
\def\Sg#1{{\sf Sg}(#1)}
\def\Sgo#1{\Sgg\zA{#1}}
\def\Sgn#1{\Sgg{\zA'}{#1}}
\def\Sgg#1#2{{\sf Sg}_{#1}(#2)}
\def\Cg#1{{\sf Cg}(#1)}
\def\Cgg#1#2{{\sf Cg}_{#1}(#2)}
\def\tol{{\sf tol}}
\def\lnk{{\sf lk}}
\def\rbcomp#1{{\sf rbcomp}(#1)}
%%%%%%%%%%%%%%%%%%%%%%%%%%%%%%%%%%
% Operations
\let\cd=\cdot
\let\eq=\equiv
\let\op=\oplus
\let\omn=\ominus
\let\meet=\wedge
\let\join=\vee
\let\tm=\times
\def\ldiv{\mathbin{\backslash}}
\def\rdiv{\mathbin/}
%%%%%%%%%%%%%%%%%%%%%%%%%%%%%%%%%%
% Tame congruence
\def\typ{{\sf typ}}
\def\zz{{\un 0}}
\def\zo{{\un 1}}
\def\one{{\bf1}}
\def\two{{\bf2}}
\def\three{{\bf3}}
\def\four{{\bf4}}
\def\five{{\bf5}}
\def\pq#1{(\zz_{#1},\mu_{#1})}
%%%%%%%%%%%%%%%%%%%%%%%%%%%%%%%%%%%
% Accents and so
\let\wh=\widehat
\def\ox{\ov x}
\def\oy{\ov y}
\def\oz{\ov z}
\def\of{\ov f}
\def\oa{\ov a}
\def\ob{\ov b}
\def\oc{\ov c}
\def\od{\ov d}
\def\oob{\ov{\ov b}}
\def\rx{{\rm x}}
\def\rf{{\rm f}}
\def\rrm{{\rm m}}
\let\un=\underline
\let\ov=\overline
\let\cc=\circ
\let\rb=\diamond
\def\ta{{\tilde a}}
\def\tz{{\tilde z}}
%%%%%%%%%%%%%%%%%%%%%%%%%%%%%%%%%%%%%%%%%%
% Abbreviations for algebras and clones
\def\zZ{{\mathbb Z}}
\def\B{{\mathcal B}}
\def\P{{\mathcal P}}
\def\zA{{\mathbb A}}
\def\zB{{\mathbb B}}
\def\zC{{\mathbb C}}
\def\zD{{\mathbb D}}
\def\zE{{\mathbb E}}
\def\zF{{\mathbb F}}
\def\zG{{\mathbb G}}
\def\zK{{\mathbb K}}
\def\zL{{\mathbb L}}
\def\zM{{\mathbb M}}
\def\zN{{\mathbb N}}
\def\zQ{{\mathbb Q}}
\def\zR{{\mathbb R}}
\def\zS{{\mathbb S}}
\def\zT{{\mathbb T}}
\def\zW{{\mathbb W}}
\def\bK{{\bf K}}
\def\C{{\bf C}}
\def\M{{\bf M}}
\def\E{{\bf E}}
\def\N{{\bf N}}
\def\O{{\bf O}}
\def\bN{{\bf N}}
\def\bX{{\bf X}}
\def\GF{{\rm GF}}
\def\cC{{\mathcal C}}
\def\cA{{\mathcal A}}
\def\cB{{\mathcal B}}
\def\cD{{\mathcal D}}
\def\cE{{\mathcal E}}
\def\cF{{\mathcal F}}
\def\cG{{\mathcal G}}
\def\cH{{\mathcal H}}
\def\cI{{\mathcal I}}
\def\cK{{\mathcal K}}
\def\cL{{\mathcal L}}
\def\cP{{\mathcal P}} 
\def\cR{{\mathcal R}}
\def\cRY{{\mathcal RY}}
\def\cS{{\mathcal S}}
\def\cT{{\mathcal T}}
\def\oB{{\ov B}}
\def\oC{{\ov C}}
\def\ooB{{\ov{\ov B}}}
\def\ozB{{\ov{\zB}}}
\def\ozD{{\ov{\zD}}}
\def\ozG{{\ov{\zG}}}
\def\tcA{{\widetilde\cA}}
\def\tcC{{\widetilde\cC}}
\def\tcF{{\widetilde\cF}}
\def\tcI{{\widetilde\cI}}
\def\tB{{\widetilde B}}
\def\tC{{\widetilde C}}
\def\tD{{\widetilde D}}
\def\ttB{{\widetilde{\widetilde B}}}
\def\ttC{{\widetilde{\widetilde C}}}
\def\tba{{\tilde\ba}}
\def\ttba{{\tilde{\tilde\ba}}}
\def\tbb{{\tilde\bb}}
\def\ttbb{{\tilde{\tilde\bb}}}
\def\tbc{{\tilde\bc}}
\def\tbd{{\tilde\bd}}
\def\tbe{{\tilde\be}}
\def\tbt{{\tilde\bt}}
\def\tbu{{\tilde\bu}}
\def\tbv{{\tilde\bv}}
\def\tbw{{\tilde\bw}}
\def\tdl{{\tilde\dl}}
\def\ocP{{\ov\cP}}
\def\tzA{{\widetilde\zA}}
\def\tzC{{\widetilde\zC}}
\def\new{{\mbox{\footnotesize new}}}
\def\old{{\mbox{\footnotesize old}}}
\def\prev{{\mbox{\footnotesize prev}}}
\def\oo{{\mbox{\sf\footnotesize o}}}
\def\pp{{\mbox{\sf\footnotesize p}}}
\def\nn{{\mbox{\sf\footnotesize n}}}
\def\oR{{\ov R}}
\def\bA{\mathbf{R}}
%%%%%%%%%%%%%%%%%%%%%%%%%%%%%%%%%%%%%%%
% Abbreviations for varieties
\def\gA{{\mathfrak A}}
\def\gV{{\mathfrak V}}
\def\gS{{\mathfrak S}}
\def\gK{{\mathfrak K}}
\def\gH{{\mathfrak H}}
%%%%%%%%%%%%%%%%%%%%%%%%%%%%%%%%%%%%%%%%
% Vectors
\def\ba{{\bf a}}
\def\bb{{\bf b}}
\def\bc{{\bf c}}
\def\bd{{\bf d}}
\def\be{{\bf e}}
\def\bbf{{\bf f}}
\def\bg{{\bf g}}
\def\bh{{\bf h}}
\def\bi{{\bf i}}
\def\bm{{\bf m}}
\def\bo{{\bf o}}
\def\bp{{\bf p}}
\def\bs{{\bf s}}
\def\bu{{\bf u}}
\def\bt{{\bf t}}
\def\bv{{\bf v}}
\def\bx{{\bf x}}
\def\by{{\bf y}}
\def\bw{{\bf w}}
\def\bz{{\bf z}}
\def\ga{{\mathfrak a}}
\def\oal{{\ov\al}}
\def\obeta{{\ov\beta}}
\def\ogm{{\ov\gm}}
\def\oep{{\ov\varepsilon}}
\def\oeta{{\ov\eta}}
\def\oth{{\ov\th}}
\def\ovm{{\ov\mu}}
\def\ozero{{\ov0}}
%%%%%%%%%%%%%%%%%%%%%%%%%%%%%%%%%%%%%%%%
% Constraint satisfaction Problem
\def\CCSP{\hbox{\rm c-CSP}}
\def\CSP{{\rm CSP}}
\def\NCSP{{\rm \#CSP}}
\def\mCSP{{\rm MCSP}}
\def\FP{{\rm FP}}
\def\PTIME{{\bf PTIME}}
\def\GS{\hbox{($*$)}}
\def\ry{\hbox{\rm r+y}}
\def\rb{\hbox{\rm r+b}}
\def\Gr#1{{\mathrm{Gr}(#1)}}
\def\Grp#1{{\mathrm{Gr'}(#1)}}
\def\Grpr#1{{\mathrm{Gr''}(#1)}}
\def\Scc#1{{\mathrm{Scc}(#1)}}
\def\rel{R}
\def\relo{Q}
\def\rela{S}
\def\dep{\mathsf{dep}}
\def\Filt{\mathrm{Ft}}
\def\Filts{\mathrm{Fts}}
\def\Agr{$\mathbb{A}$}
\def\Al{\mathrm{Alg}}
\def\Sig{\mathrm{Sig}}
\def\strat{\mathsf{strat}}
\def\relmax{\mathsf{relmax}}
\def\srelmax{\mathsf{srelmax}}
\def\Meet{\mathsf{Meet}}
\def\amax{\mathsf{amax}}
\def\umax{\mathsf{umax}}
\def\as{\mathsf{as}}
\def\star{\hbox{$(*)$}}
\def\bmal{{\mathbf m}}
\def\Af{\mathsf{Af}}
\let\sqq=\sqsubseteq

%%%%%%%%%%%%%%%%%%%%%%%%%%%%%%%%%%%%%%%%
% Mathematical abbreviations
\let\sse=\subseteq
\def\ang#1{\langle #1 \rangle}
\def\angg#1{\left\langle #1 \right\rangle}
\def\dang#1{\ang{\ang{#1}}}
\def\vc#1#2{#1 _1\zd #1 _{#2}}
\def\tms#1#2{#1 _1\tm\dots\tm #1 _{#2}}
\def\zd{,\ldots,}
\let\bks=\backslash
\def\red#1{\vrule height7pt depth3pt width.4pt
\lower3pt\hbox{$\scriptstyle #1$}}
\def\fac#1{/\lower2pt\hbox{$\scriptstyle #1$}}
\def\me{\stackrel{\mu}{\eq}}
\def\nme{\stackrel{\mu}{\not\eq}}
\def\eqc#1{\stackrel{#1}{\eq}}
\def\cl#1#2{\arraycolsep0pt
\left(\begin{array}{c} #1\\ #2 \end{array}\right)}
\def\cll#1#2#3{\arraycolsep0pt \left(\begin{array}{c} #1\\ #2\\
#3 \end{array}\right)}
\def\clll#1#2#3#4{\arraycolsep0pt
\left(\begin{array}{c} #1\\ #2\\ #3\\ #4 \end{array}\right)}
\def\cllll#1#2#3#4#5#6{ \left(\begin{array}{c} #1\\ #2\\ #3\\
#4\\ #5\\ #6 \end{array}\right)}
\def\pr{{\rm pr}}
\let\upr=\uparrow
\def\ua#1{\hskip-1.7mm\uparrow^{#1}}
\def\sua#1{\hskip-0.2mm\scriptsize\uparrow^{#1}}
\def\lcm{{\rm lcm}}
\def\perm#1#2#3{\left(\begin{array}{ccc} 1&2&3\\ #1 
\end{array}\right)}
\def\w{$\wedge$}
\let\ex=\exists
\def\NS{{\sc (No-G-Set)}}
\def\lev{{\sf lev}}
\let\rle=\sqsubseteq
\def\ryle{\le_{ry}}
\def\ryprec{\le_{ry}}
\def\os{\mbox{[}}
\def\zs{\mbox{]}}
\def\link{{\sf link}}
\def\solv{\stackrel{s}{\sim}}
\def\mal{\mathbf{m}}
\def\precs{\prec_{as}}
\def\maj{\mathrm{maj}}
\let\dg=\dagger

%%%%%%%%%%%%%%%%%%%%%%%%%%%%%%%%%%%%%%%%%%%
% Other abbreviations
\def\lb{$\linebreak$}
%%%%%%%%%%%%%%%%%%%%%%%%%%%%%%%%%%%%%%%%%%%
% Functions
\def\ar{\hbox{ar}}
\def\Im{{\sf Im}\;}
\def\deg{{\sf deg}}
\def\id{{\rm id}}
%%%%%%%%%%%%%%%%%%%%%%%%%%%%%%%%%%%%%%%%%%
% Greek symbols
\let\al=\alpha
\let\gm=\gamma
\let\dl=\delta
\let\ve=\varepsilon
\let\ld=\lambda
\let\om=\omega
\let\vf=\varphi
\let\vr=\varrho
\let\th=\theta
\let\sg=\sigma
\let\Gm=\Gamma
\let\Dl=\Delta
%%%%%%%%%%%%%%%%%%%%%%%%%%%%%%%%%%%%%%%%%%%
% Fonts and special symbols
\font\tengoth=eufm10 scaled 1200
\font\sixgoth=eufm6
\def\goth{\fam12}
\textfont12=\tengoth
\scriptfont12=\sixgoth
\scriptscriptfont12=\sixgoth
\font\tenbur=msbm10
\font\eightbur=msbm8
\def\bur{\fam13}
\textfont11=\tenbur
\scriptfont11=\eightbur
\scriptscriptfont11=\eightbur
\font\twelvebur=msbm10 scaled 1200
\textfont13=\twelvebur
\scriptfont13=\tenbur
\scriptscriptfont13=\eightbur
\mathchardef\nat="0B4E
\mathchardef\eps="0D3F
%%%%%%%%%%%%%%%%%%%%%%%%%%%%%%%%%%%%%%%%%%%%%
%%%%%%%%%%%%%%%%%%%%%%%%%%%%%%%%%%%%%%%%%%%%%
\title[Graphs of finite algebras I]{Graphs of finite algebras: edges, and connectivity}
\corrauthor[A. A. Bulatov]{Andrei A.\ Bulatov}
\address{School of Computing Science, Simon Fraser University, Burnaby, Canada}
\urladdr{www.cs.sfu.ca/~abulatov}
\email{abulatov@cs.sfu.ca}
\thanks{This work was supported by an NSERC Discovery grant. }

\begin{abstract}
We refine and advance the study of the local structure of idempotent finite 
algebras started in [A.Bulatov, \emph{The Graph of a Relational Structure 
and Constraint Satisfaction Problems}, LICS, 2004]. We introduce a
graph-like structure on an arbitrary finite idempotent algebra including those 
admitting type \one. We show that this graph is connected, its edges can 
be classified into 4 types corresponding to the local behavior (set, semilattice, 
majority, or affine) of certain term operations. We also show that if the variety 
generated by the algebra omits type \one, then the structure of the algebra 
can be `improved' without introducing type \one\ by choosing
an appropriate reduct of the original algebra. Taylor minimal idempotent 
algebras introduced recently are a special case of such reducts. Then we 
refine this structure demonstrating that the edges of the graph of an algebra 
omitting type \one\ can be made `thin', 
that is, there are term operations that behave very similar to semilattice, 
majority, or affine operations on 2-element subsets of the algebra. Finally, 
we prove certain connectivity properties of the refined structures.

This research is motivated by the study of the Constraint Satisfaction Problem, 
although the problem itself does not really show up in this paper.
\end{abstract}
\maketitle

%%%%%%%%%%%%%%%%%%%%%%%%%%%%%%%%%%%%
%%%%%%%%%%%%%%%%%%%%%%%%%%%%%%%%%%%%
\section{Introduction}\label{sec:introduction}

% about the 2004 paper
The study of the Constraint Satisfaction Problem (CSP) and especially the 
Dichotomy Conjecture triggered a wave of research in universal algebra, as it 
turns out that the algebraic approach to the CSP developed in 
\cite{Bulatov05:classifying,Jeavons97:closure}
is the most prolific one in this area. These developments have led to a number 
of strong results about the CSP, see, e.g., 
\cite{Barto11:conservative,Barto14:local,Barto12:near,Bulatov06:3-element,% 
Bulatov11:conservative,Bulatov14:conservative,Bulatov06:simple,%
Idziak10:few}, and ultimately the resolution of the Dichotomy Conjecture
in \cite{Bulatov17:dichotomy,Zhuk17:proof,Zhuk18:modification}. 
However, successful application of the algebraic approach also 
requires new results about the structure of finite algebras. Two ways to describe 
this structure have been proposed. One is based on absorption properties 
\cite{Barto15:constraint,Barto12:absorbing} and has led not only to new 
results on the CSP, but also to significant developments in universal algebra itself.

In this paper we refine and advance the alternative approach originally 
developed in \cite{Bulatov04:graph,Bulatov11:conjecture,Bulatov08:recent}, 
which is based on the local structure of finite algebras. This
approach identifies subalgebras or factors of an algebra having `good' term 
operations, that is, operations of one of the three types: semilattice, majority, 
or affine. It then explores the graph or hypergraph formed by such subalgebras, 
and exploits its connectivity properties. In a nutshell, this method stems from 
the early study of the CSP over so called conservative algebras 
\cite{Bulatov11:conservative}, and has led to a much simpler proof of the 
Dichotomy Conjecture for conservative algebras \cite{Bulatov16:conservative} 
and to a characterization of CSPs solvable by consistency algorithms 
\cite{Bulatov09:bounded}. In spite of these applications the original method 
suffer from a number of drawbacks that make its use difficult. 

In the present paper we refine many of the constructions, generalize them to 
include algebras admitting type \one, and fix the deficiencies of the original
method. Similar to \cite{Bulatov04:graph,Bulatov08:recent} an edge is a pair of 
elements $a,b$ such that there is a factor algebra of the subalgebra generated 
by $a,b$ that has an operation which is semilattice, majority, or affine on the 
blocks containing $a,b$. Here we extend the definition of an edge by allowing 
edges to have the unary type, when the corresponding factor is a set. These
operations, or lack thereof, determine the type of edge $ab$. In this paper 
we also allow edges to have more than one type if there are several factors 
witnessing different types. 

The main difference from the previous results is the 
introduction of oriented \emph{thin} majority and affine edges.
An edge $ab$ is said to be thin if there is a term operation that is semilattice 
on $\{a,b\}$, or there is a term operation that satisfies the identities of a 
majority or affine term on $\{a,b\}$, but only for some choices of values of 
the variables. In particular, for a thin majority edge we require the 
existence of a term operation $m$ such that $m(a,b,b)=m(b,a,b)=m(b,b,a)=b$,
and for a thin affine edge we require that there is a term operation $d$ 
with $d(b,a,a)=d(a,a,b)=b$, but $m,d$ do not have to satisfy similar conditions when $a,b$ are swapped. Oriented thin edges allow us to prove a stronger 
version of the connectivity of the graph related to an algebra.
This updated approach makes it possible to give a much simpler proof of the
result of \cite{Bulatov09:bounded} (see also \cite{Barto14:local}), and, 
eventually, proving the Dichotomy Conjecture 
\cite{Bulatov17:dichotomy,Bulatov17:dichotomy-full}, see also
\cite{Zhuk17:proof,Zhuk17:proof-arxiv,Zhuk18:modification}, however, 
this is a subject of subsequent papers.

%%%%%%%%%%%%%%%%%%%%%%%%%%%%%%%%%%%
%%%%%%%%%%%%%%%%%%%%%%%%%%%%%%%%%%%
\section{Preliminaries}\label{sec:preliminaries}

% use the standard books, TCT
In terminology and notation we mainly follow the standard texts on
universal algebra \cite{Burris81:universal,Mckenzie87:algebras}.
We also assume some familiarity with the basics of the tame
congruence theory \cite{Hobby88:structure}, although its use is
restricted to some proofs in the first two sections.
All algebras in this paper are assumed to be finite and idempotent.

% notation: tuples, relations, projections, subalgebras, operations
Algebras will be denoted by $\zA,\zB$, etc. The subalgebra of an
algebra $\zA$ generated by a set $B\sse\zA$ is denoted $\Sgg\zA B$,
or if $\zA$ is clear from the context simply by $\Sg B$. The set
of term operations of algebra $\zA$ is denoted by $\Term(\zA)$.
Subalgebras of direct products are often considered as relations.
An element (a tuple) from $\tms\zA n$ is denoted in boldface, say,
$\ba$, and its $i$th component is referred to as
$\ba[i]$, that is, $\ba=(\ba[1]\zd\ba[n])$. The set $\{1\zd n\}$ will be
denoted by $[n]$. For $I\sse[n]$, say, $I=\{\vc ik\}$, $i_1<\dots<i_k$,
by $\pr_I\ba$ we denote the $k$-tuple $(\ba[i_1]\zd\ba[i_k])$, and for
$\rel\sse\tms\zA n$ by $\pr_I\rel$ we denote the set
$\{\pr_I\ba\mid\ba\in\rel\}$. If $I=\{i\}$ or $I=\{i,j\}$ we write
$\pr_i,\pr_{ij}$ rather than $\pr_I$. The tuple $\pr_I\ba$ and
relation $\pr_I\rel$ are called the \emph{projections} of, respectively,
$\ba$ and $\rel$ on $I$.
A subalgebra (a relation) $\rel$ of $\tms\zA n$ is said to be a
\emph{subdirect product} of $\vc\zA n$ if $\pr_i\rel=\zA_i$
for every $i\in[n]$. For a congruence $\al$ of $\zA$, $a\in\zA$, and term
operation $f$ of $\zA$, by $a\fac\al$ we denote the $\al$-block containing
$a$, by $\zA\fac\al$ the factor algebra modulo $\al$, and by $f\fac\al$
the factor operation of $f$, that is, the operation on $\zA\fac\al$ given
by $f\fac\al(a_1\fac\al\zd a_n\fac\al)=f(\vc an)\fac\al$. For $B\sse\zA^2$,
the congruence generated by $B$ will be denoted by $\Cgg\zA B$
or just $\Cg B$. By $\zz_\zA,\zo_\zA$ we denote the least (i.e., the
equality relation), and the greatest (i.e., the total relation)
congruences of $\zA$, respectively. Again, we often simplify this
notation to $\zz,\zo$.

An algebra is called a \emph{set} if all of its term operations are
projections.

Recall that a {\em tolerance} of $\zA$ is a binary reflexive and
symmetric relation compatible with $\zA$. The transitive closure
of a tolerance is a congruence of $\zA$. Let $\al,\beta$ be
congruences of $\zA$ such that $\al\prec\beta$, that is,
$\al<\beta$ and $(\al,\beta)$ is a prime interval. Then a
\emph{$(\al,\beta)$-minimal set} $U$ is a subset of $\zA$, minimal
with respect to inclusion among subsets of the form $U=f(\zA)$, for
some unary polynomial $f$ of $\zA$ such that $f(\beta)\not\sse\al$.
If this is the case then $f$ can also be chosen
to be idempotent. The intersection of $U$ with a $\beta$-block
is a \emph{trace} if it overlaps with more than one $\al$-block.
Any two $(\al,\beta)$-minimal sets $U,V$ are \emph{polynomially
isomorphic}, that is, there are unary polynomials $g,h$ such that
$g(U)=V$, $h(V)=U$, and $g^{-1}\circ h, h^{-1}\circ g$ are identity
transforations on $V$ and $U$, respectively. The prime factor 
$(\al,\beta)$ has type
\one,\two,\three,\four, or \five\ if every $(\al,\beta)$-trace
$U\fac\al$ is polynomially equivalent to an algebra with only unary operations,
1-dimensional vector space, 2-element Boolean algebra, lattice,
or semilattice, respectively. The type of $(\al,\beta)$ is denoted
by $\typ(\al,\beta)$.

If $\al$ is a congruence
of an algebra $\zA$ and $\rel$ is a compatible binary relation,
then the {\em $\al$-closure} of $\rel$ is
defined to be $\al\circ\rel\circ\al$. A relation equal to its
$\al$-closure is said to be {\em $\al$-closed}. If $(\al,\beta)$ is a prime
factor of $\zA$, then the {\em basic tolerance for} $(\al,\beta)$ (see
\cite{Hobby88:structure}, Chapter~5) is the $\al$-closure of the
relation $\al\cup\bigcup\{N^2\mid N$ is an $(\al,\beta)$-trace$\}$
if $\typ(\al,\beta)\in\{\two,\three\}$, and it is the
$\al$-closure of the compatible relation generated by
$\al\cup\bigcup\{N^2\mid N$ is an $(\al,\beta)$-trace$\}$ if
$\typ(\al,\beta)\in\{\four,\five\}$. The basic tolerance is the
smallest $\al$-closed tolerance $\tau$ of $\zA$ such that
$\al\ne\tau\sse\beta$.

Let $(\al,\beta)$ be a prime factor of $\zA$. An {\em
$(\al,\beta)$-quasi-order} is a compatible reflexive and
transitive relation $\rel$ such that $\rel\cap\rel^{-1}=\al$, and the
transitive closure of $\rel\cup\rel^{-1}$ is $\beta$. The quotient
$(\al,\beta)$ is said to be {\em orderable} if there exists an
$(\al,\beta)$-quasi-order. By Theorem~5.26 of
\cite{Hobby88:structure}, if $\typ(\al,\beta)\ne\one$, then $(\al,\beta)$ is
orderable if and only if $\typ(\al,\beta)\in\{\four,\five\}$.

%%%%%%%%%%%%%%%%%%%%%%%%%%%%%%%%%%%%
%%%%%%%%%%%%%%%%%%%%%%%%%%%%%%%%%%%%
\section{Graph: Thick edges}\label{sec:thick}

After recalling several basic facts about idempotent algebras, we start 
with introducing `thick' edges, one of the main constructions of this paper.

%%%%%%%%%%%%%%%%%%%%%%%%%%%%%%%%
\subsection{Basic facts about idempotent algebras}

An algebra $\zA=(A;F)$ is \emph{idempotent} if $f(x\zd x)=x$
for every $f\in F$. The algebra $\zA'=(A,F')$ where $F'$ is the set 
of all idempotent operations from $\Term(\zA)$ is said to be the 
\emph{full idempotent reduct} of $\zA$. Let $\tau$ be a tolerance. A set
$B\sse\zA$ maximal with respect to inclusion and such that $B^2\sse
\tau$ is said to be a {\em class} of $\tau$. We will need the following
simple observation.
\begin{lemma}\label{lem:tol-class}
Every class of a tolerance of an idempotent algebra is a subalgebra.
\end{lemma}

\begin{proof}
Let $\tau$ be a tolerance of an idempotent algebra $\zA$ and $B$ its
class. Then, for any $a\not\in B$, there is $b\in B$ such that
$(a,b)\not\in\tau$. If $B$ is not a subuniverse, then, for a certain
term operation $f$ of $A$ and $\vc bn\in B$, we have
$a=f(\vc bn)\not\in B$. Take $b\in B$ such that $(a,b)\not\in\tau$.
Since $(b_1,b)\zd(b_n,b)\in\tau$, we get
$(f(\vc bn),f(b\zd b))=(a,b)\in\tau$, a contradiction.
\end{proof}

Let $G=(V,E)$ be a hypergraph. A {\em path} in $G$
is a sequence $\vc Hk$ of hyperedges such that
$H_i\cap H_{i+1}\ne\eps$, for $1\le i< k$. The hypergraph $G$
is said to be {\em connected} if, for any $a,b\in V$, there is a
path $\vc Hk$ such that $a\in H_1$, $b\in H_k$.

The universe of an algebra $\zA$ along with the family of all
of its proper subalgebras form a hypergraph denoted by
$\cH(\zA)$. Lemma~\ref{lem:tol-class} implies that, for a simple
idempotent algebra $\zA$, the hypergraph $\cH(\zA)$ is connected
unless $\zA$ is tolerance free. In the latter case it can be disconnected.

Recall that an element $a$ of an algebra $\zA$
is said to be {\em absorbing} if whenever $t(x,\vc yn)$ is an
$(n+1)$-ary term operation of $\zA$ such that $t$ depends on $x$ and
$\vc bn\in\zA$, then $t(a,\vc bn)=a$. A congruence $\th$ of
$\zA^2$ is said to be {\em skew} if it not a product congruence of the form $\th_1\tm\th_2=\{((a_1,a_2),(b_1,b_2))\mid (a_1,b_1)\in\th_1, (a_2,b_2)\in\th_2\}$ for $\th_1,\th_2\in\Con(\zA)$. In particular, if $\zA$ is simple it happens if it is the kernel of no projection
mapping of $\zA^2$ onto its factors. The following theorem due to
Kearnes \cite{Kearnes96:idempotent} provides some information
about the structure of simple idempotent algebras.

\begin{theorem}[\cite{Kearnes96:idempotent}]\label{the:simple}
Let $\zA$ be a simple idempotent algebra. Then exactly
one of the following holds:
\begin{itemize}
\item[(a)]
$\zA$ is abelian and is term equivalent to a subalgebra of a reduct of
a module;
\item[(b)]
$\zA$ has an absorbing element;
\item[(c)]
$\zA^2$ has no skew congruence.
\end{itemize}
\end{theorem}

Since $\zA$ is idempotent, in most cases in item (a) rather than working 
with a module $\zM$ we use its full idempotent reduct. Slightly abusing the 
terminology we will call such a reduct simply by a module.

Option (a) of Theorem~\ref{the:simple} can be made more precise. In
\cite{Valeriote90:finite} Valeriote proved that every finite simple abelian
idempotent algebra is strictly simple, that is, does not have proper 
subalgebras containing more than one element. It then follows from
the result of Szendrei \cite{Szendrei87:simple} that every such algebra
is either a set or a module. 

\begin{prop}\label{pro:simple-abelian}
Let $\zA$ be a simple finite idempotent Abelian algebra. Then $\zA$
is either a set, or a module.
\end{prop}

%%%%%%%%%%%%%%%%%%%%%%%%%%%%%%%%%%%
\subsection{The four types of edges}\label{sec:three-types}

Let $\zA$ be an algebra with universe $A$.
We introduce graph $\cG(\zA)$ as follows: The vertex set of $\cG(\zA)$
is the set $A$. A pair $ab$ of vertices is an \emph{edge} if and only if
there exists a congruence $\th$ of $\Sg{a,b}$ such that either
$\Sg{a,b}\fac\th$ is a set, or there is a term operation $f$ of $\zA$ such
that either $f\fac\th$ is an affine operation on $\Sg{a,b}\fac\th$, or
$f\fac\th$ is a semilattice operation on $\{a\fac\th,b\fac\th\}$, or $f\fac\th$
is a majority operation on $\{a\fac\th,b\fac\th\}$.

If there exists a congruence $\th$ and a term operation $f$ of $\zA$ such that
$f\fac\th$ is a semilattice operation on $\{a\fac\th,b\fac\th\}$ then $ab$ is
said to have the {\em semilattice type}. Edge $ab$ is of the
{\em majority type} if there are a congruence $\th$ and $f\in\Term(\zA)$
such that $f\fac\th$ is a majority operation on $\{a\fac\th,b\fac\th\}$,
but no operation is semilattice on $\{a\fac\th,b\fac\th\}$. Pair $ab$ has 
the {\em affine type}
if there are a congruence $\th$ and $f\in\Term(\zA)$
such that $\Sg{a,b}\fac\th$ is a module and $f\fac\th$ is its 
affine operation $x-y+z$. Finally, $ab$ is of the
\emph{unary type} if $\Sg{a,b}\fac\th$ is a set (in which case $\Sg{a,b}\fac\th=\{a\fac\th,b\fac\th\}$). In all cases we say that
congruence $\th$ and operation $f$ \emph{witness} the type of edge $ab$. 
As is easily seen, congruence $\th$ can always be chosen to be a maximal congruence of $\Sg{a,b}$ if $ab$ is a unary, semilattice, or affine edge. For majority edges the situation is more complicated as the following example shows. Let $\zA =
(\{a, b, 0\}; f,\meet)$. The operation $\meet$ is the semilattice operation which satisfies $a\meet b =
0$. The operation $f$ is the majority operation on $\{a,b\}$, but if $0\in\{x,y,z\}$, then
$f(x, y, z) = 0$. Clearly, $ab$ is a majority edge with respect to the identity
relation, but with respect to any of the two maximal congruences of $\zA=\Sg{a,b}$ (partitioning $A$ into $\{a,0\},\{b\}$ and $\{a\},\{b,0\}$, respectively), $ab$ is a semilattice edge. 

Note that, for every edge $ab$ of $\cG(\zA)$, there is an associated
pair $a\fac\th,b\fac\th$ from the factor subalgebra. We will need both of these
types of pairs and will sometimes call $a\fac\th,b\fac\th$ a {\em thick} edge
(see Fig.~\ref{fig:edges}). There may be many congruences certifying 
the type of the edge $ab$. In most cases it does not make much difference
which congruence to choose, so we will use just any 
congruence certifying the type of $ab$.
\begin{figure}[t]
\centerline{\includegraphics[scale=.3]{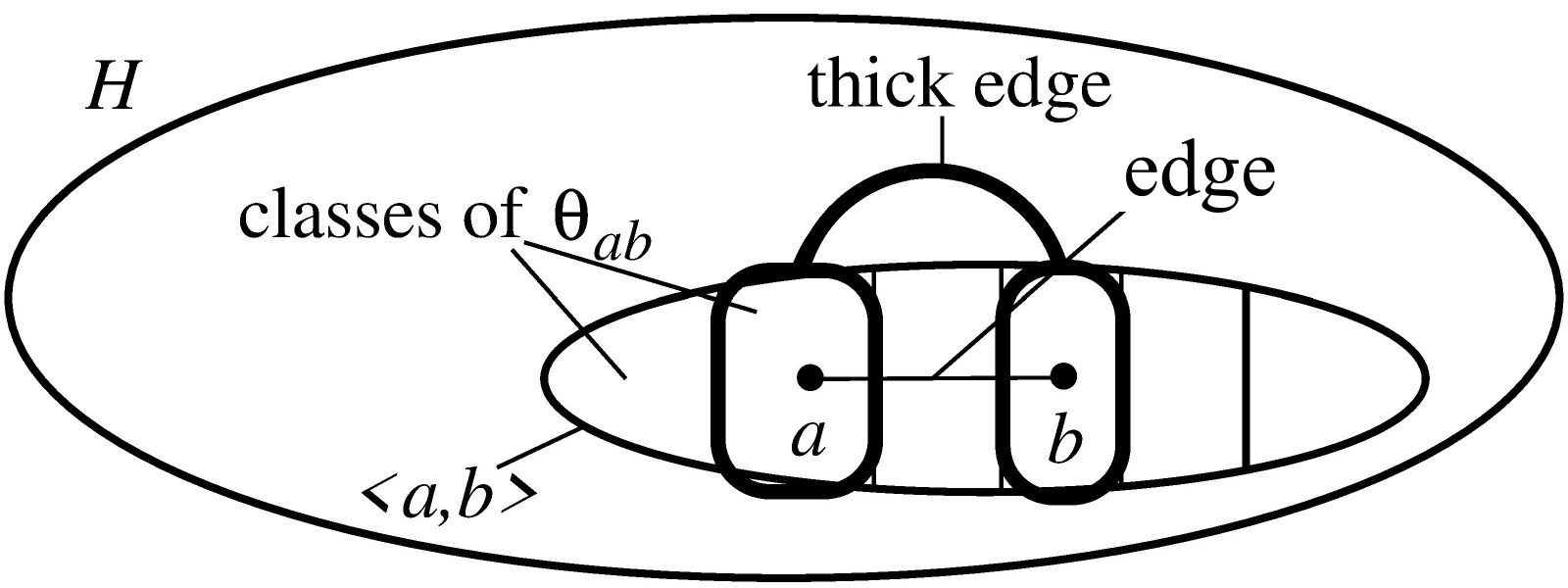}}
\caption{Edges and thick edges}\label{fig:edges}
\end{figure}
Note also that a pair $ab$ may have more than one type witnessed by
different congruences $\th$. 

\begin{example}\label{exa:no-edge}
Let $A=\{a,b,c\}$ and binary operations $f,g$ on $A$ defined 
as follows: $f(a,c)=f(c,a)=g(a,c)=g(c,a)=c$, $f(b,c)=f(c,b)=g(b,c)=g(c,b)=c$, 
$f(a,b)=g(b,a)=c$, and $f(b,a)=b$, $g(a,b)=a$.
Let $\zA=(A,f,g)$, that is, $\zA$ is an `almost' 3-element semilattice with 
respect to each of $f,g$, in which $a,b$ are incomparable, but unlike a
3-element semilattice $\zA$ is simple. As is easily seen, each of the pairs 
$\{a,c\},\{b,c\}$ generates a 2-element subalgebra term equivalent
to a 2-element semilattice, and therefore is an edge of $\zA$
of the semilattice type. On the other hand, $\{a,b\}$ generates $\zA$,
and it can be easily seen that there is no term operation of $\zA$
that is semilattice, majority, or affine on $\{a,b\}$. Therefore $ab$ is 
not an edge.
\end{example}

%%%%%%%%%%%%%%%%%%%%%%%%%%%%%%%%%%%
\subsection{General connectivity}\label{sec:general-connectivity}

The main result of this section is the connectedness of the graph
$\cG(\zA)$ for arbitrary idempotent algebras, and the connection
between omitting types \one\ and \two\ and avoiding edges of certain types.
Recall that for a class $\gA$ of algebras $\H(\gA)$ and $\S(\gA)$
denote the class of all homomorphic images, and the class of
all subalgebras of algebras from $\gA$, respectively. More precisely, we prove
the following

\begin{theorem}\label{the:connectedness}
Let $\zA$ be an idempotent algebra. Then
\begin{itemize}
\item[(1)]
$\cG(\zA)$ is connected;
\item[(2)]
$\H\S(\zA)$ omits type \one\ if and only if $\cG(\zA)$ contains no
edges of the unary type;
\item[(3)]
$\H\S(\zA)$ omits types \one\ and \two\ if and only if $\cG(\zA)$
contains no edges of the unary and affine types.
\end{itemize}
\end{theorem}

%
% proposition on simple algebras
First we prove Theorem~\ref{the:connectedness} for simple algebras.
We will need the following easy observation.
\begin{lemma}\label{lem:binary-tol}
Let $\rel$ be an $n$-ary compatible relation on $\zA$ such that, for
any $i\in[n]$, $\pr_i\rel=\zA$. Then, for any $i\in[n]$,
the relation $\tol_i(\rel)=\{(a,b)\mid $ there are $a_1\zd a_{i-1},a_{i+1}\zd
a_n\in\zA$ such that $(a_1\zd a_{i-1},a,a_{i+1}\zd a_n),\lb(a_1\zd
a_{i-1},b,a_{i+1}\zd a_n)\in\rel\}$ is a tolerance of $\zA$.
\end{lemma}

Tolerance of the form $\tol_i(\rel)$ will be called \emph{link tolerance}, or
\emph{$i$th link tolerance}. We will drop mention of $\rel$ whenever it does not cause confusion.

\begin{prop}\label{pro:simple}
Let $\zA$ be a simple idempotent algebra. Then one of the following
holds:
\begin{itemize}
\item[(1)]
$\zA$ is a 2-element set;
\item[(2)]
$\zA$ is term equivalent to a module;
\item[(3)]
$\cH(\zA)$ is connected;
\item[(4)]
$\zA=\Sg{a,b}$ for some $a,b\in\zA$, and for any such $a,b$ either $a,b$ are connected in $\cH(\zA)$, or there is a binary term operation $f$ such that $f$
is a semilattice operation on $\{a,b\}$, or there is a ternary term operation $g$ and an automorphism of $\zA$ such that $g$ is a majority operation
on $\{a,b\}$ and the automorphism  swaps $a$ and~$b$.
\end{itemize}
\end{prop}

\begin{proof}
We consider the three options given in Theorem~\ref{the:simple}.\\[2mm]
(a) Suppose first that $\zA$ is abelian. Then by 
Proposition~\ref{pro:simple-abelian}
it is either a set or a module. As $\zA$
is simple, in the former case it is also 2-element, and we are in case (1) of 
Proposition~\ref{pro:simple}. In the latter case we of course obtain 
case (2) of the proposition. In the rest of the proof
we assume $\zA$ non-abelian, in particular the factor
$\zz_\zA\prec\zo_\zA$ has type \three,\four, or \five.\\[2mm]
(b) Next assume that $\zA$ has an absorbing element, and let $a$ be the 
absorbing element of $\zA$. Then for any
$b\in\zA$ the set $\{a,b\}$ is a subalgebra. Indeed, let
$f(\vc xk)$ be a term operation of $\zA$ (for simplicity assume
that it depends on all of its arguments). Then for $\vc ak\in\{a,b\}$
we have $f(\vc ak)=a$ if any of the $a_i$ equals $a$, and
$f(\vc ak)=b$ otherwise, as $f$ is idempotent.
Therefore if $|\zA|\ge3$ then $\cH(\zA)$ is connected through such
2-element subalgebras, and we are in case (3) of 
Proposition~\ref{pro:simple}. If $|\zA|=2$, say, $\zA=\{a,b\}$, then
clearly $\zA=\Sg{a,b}$. Moreover for any non-unary term operation
$f(\vc xk)$ of $\zA$ for $g(x,y)=f(x,y\zd y)$ we have
$g(a,a)=g(a,b)=g(b,a)=a$ and $g(b,b)=b$; that is, $g$ is a
semilattice operation on $\zA$. We obtain case (4) of 
Proposition~\ref{pro:simple}.\\[2mm]
(c) If $\typ(\zz_\zA,\zo_\zA)\in\{\four,\five\}$, then by Theorem~5.26
of \cite{Hobby88:structure}, there exists a
$(\zz_\zA,\zo_\zA)$-quasi-order $\le$ on $\zA$, which is, clearly, just a
compatible partial order. Let $a\le b\in A$ be such that $a\le c\le
b$ implies $c=a$ or $c=b$. We claim that $\{a,b\}$ is a subalgebra of
$\zA$. Indeed, for any term operation $f(\vc xn)$ of $\zA$ and any
$\vc an\in\{a,b\}$, we have $a=f(a\zd a)\le f(\vc an)\le f(b\zd
b)=b$. If $|\zA|\ge 3$, it follows from Lemma~5.24(3) and Theorem~5.26(2) of
\cite{Hobby88:structure} that $\le$ is connected. Therefore we have
case (3) of Proposition~\ref{pro:simple}. If $|\zA|=2$, then $\zA$ has a semilattice or majority term operation, as follows from the description of clones on a 2-element set, \cite{Post41}. Moreover, if the mapping $\vf$ swapping the elements of $\zA$ is not an automorphism of $\zA$, the algebra $\zA$ has a term operation $f$ that is not self-dual (meaning precisely that $\vf$ does not commute with $f$). If this is the case, again, the description from \cite{Post41} implies that a semilattice operation is a term operation of $\zA$.

Suppose that $\typ(\zz_\zA,\zo_\zA)=\three$. If $\Sg{a,b}\ne\zA$ for
all $a,b\in\zA$ then $\cH(\zA)$ is connected by subalgebras
generated by 2-element sets. So, suppose $\zA=\Sg{a,b}$ for some 
$a,b\in\zA$.
Note that as $\zA$ is simple, the transitive closure of every one of its tolerance
different from the equality relation is the total relation. If a tolerance
satisfies this condition then we say that it is {\em connected}. Since if a 
connected tolerance $\tau$ exists, $\cH(\zA)$
is connected by the classes of $\tau$ that are subalgebras 
of $\zA$ by Lemma~\ref{lem:tol-class}, we obtain case (3) of 
Proposition~\ref{pro:simple}. Thus, we can assume that $\zA$ is tolerance-free.

We consider two cases.
\medskip

\noindent
{\sc Case 1.} There is no automorphism $\vf$ of $\zA$ such that
$\vf(a)=b$ and $\vf(b)=a$.\\[2mm]
Consider the relation $\rel$ generated by $(a,b),(b,a)$. By the
assumption made, $\rel$ is not the graph of a bijective mapping. By
Lemma~\ref{lem:binary-tol}, $\tol_1(\rel),\tol_2(\rel)$ are tolerances of $\zA$
different from the equality relation. Thus, they both are equal to the
total relation. Therefore, there is $c\in\zA$ such that
$(a,c),(b,c)\in\rel$. If both $\Sg{a,c},\Sg{b,c}$ are proper
subalgebras of $\zA$, then $a,b$ are connected in
$\cH(\zA)$. Otherwise, let, say, $\Sg{a,c}=\zA$. Since
$(b,a),(b,c)\in\rel$ and $\zA$ is idempotent, $(b,d)\in\rel$ for any
$d\in \zA$. In particular, $(b,b)\in\rel$. This means that there is a
binary term operation $f$ such that $f(a,b)=f(b,a)=b$, that is,
$f$ is semilattice on $\{a,b\}$, as required. We are in case (4) of 
Proposition~\ref{pro:simple}.
\medskip

\noindent
{\sc Case 2.} There is an automorphism $\vf$ of $\zA$ such that
$\vf(a)=b$ and $\vf(b)=a$.\\[2mm]
Consider the ternary relation $\rel$ generated by
$(a,a,b),(a,b,a),(b,a,a)$. As in the previous case, if we show that
$(a,a,a)\in\rel$, then the result follows, as we explain in detail later. Let also
$\rela=\{(c,\vf(c))\mid c\in \zA\}$ denote the graph of an automorphism
$\vf$ with $\vf(a)=b$ and $\vf(b)=a$.
\medskip

\noindent
{\sc Claim 1.} $\pr_{12}\rel=\zA\tm\zA$.\\[2mm]
Let $\relo=\pr_{12}\rel$ and $\relo'=\{(c,\vf(d))\mid
(c,d)\in\relo\}$. As $\relo'(x,y)=\ex z (\relo(x,z)\meet\rela(z,y))$, this
relation is compatible. Clearly, $\relo=\zA\tm\zA$ if and only if
$\relo'=\zA\tm\zA$. Notice that $(a,a),(b,b),(a,b)\in\relo'$. Since
$\typ(\zA)=\three$ and $\zA$ is tolerance free, every pair $c,d\in\zA$
is a trace. Therefore, again as $\typ(\zA)=\three$, there is a polynomial 
operation $g(x)$ with
$g(a)=c, g(b)=d$ and, hence, there is a term operation $f(x,y,z)$ such
that $f(a,b,x)=g(x)$. For this operation we have
$$
f\left(\cl aa,\cl bb,\cl ab\right)=\cl cd\in\relo'.
$$
Since such an operation can be found for any pair $c,d\in\zA$, this
proves the claim.

\medskip

Next we show that $\tol_3(\rel)$ cannot be the equality relation. Suppose
for contradiction that it is. Then the relation
$\th=\{((c_1,d_1),(c_2,d_2))\mid $ there is $e\in\zA$ such that
$(c_1,d_1,e),(c_2,d_2,e)\in\rel\}$ is a congruence of $\relo=\zA^2$. It
cannot be a skew congruence, hence, it is the kernel of the projection of
$\zA^2$ onto one of its factors. Without loss of generality let
$\th=\{((c_1,d_1),(c_2,d_2))\mid c_1=c_2\}$. This means that, for any
$e\in\zA$ and any $(c_1,d_1,e),(c_2,d_2,e)\in\rel$, we have
$c_1=c_2$. However, $(a,b,a),(b,a,a)\in\rel$, a contradiction. The same
argument applies when $\th=\{((c_1,d_1),(c_2,d_2))\mid d_1=d_2\}$.

Thus, $\tol_3(\rel)$ is the total relation, and there is $(c,d)\in\relo$ such
that $(c,d,a),(c,d,b)\in\rel$ which implies $\{(c,d)\}\tm\zA\sse\rel$.
\medskip

\noindent
{\sc Claim 2.} For any $(c',d')\in\pr_{12}\rel$, it holds that
$(c',d',a)\in\rel$.\\[2mm]
Take a term operation $g(x,y,z)$ such that $g(a,b,c)=c'$ and
$g(a,b,\vf^{-1}(d))=\vf^{-1}(d')$. Such an operation exists provided that
$c\ne\vf^{-1}(d)$, because every pair of elements of $\zA$ is a trace.
Then
\begin{align*}
\lefteqn{g\left(\cll aba,\cll baa,\cll cda\right)=\cll{g(a,b,c)}%
{g(b,a,d)}a}\\
&=\cll{c'}{\vf(g(\vf^{-1}(b),\vf^{-1}(a),\vf^{-1}(d)))}a\\
&=\cll{c'}{\vf(g(a,b,\vf^{-1}(d)))}a=
\cll{c'}{\vf(\vf^{-1}(d'))}a=\cll{c'}{d'}a\in\rel.
\end{align*}
What is left is to show that there are $c',d'$ such that $(c',d',a)\in\rel$ 
and $c'\ne\vf^{-1}(d')$. Suppose $c=\vf^{-1}(d)$. If 
$\Sg{a,c},\Sg{c,b}\ne\zA$, then $a,b$ are connected in
$\cH(\zA)$. Let $\Sg{a,c}=\zA$, and $h$ such that $h(a,c)=b$. Since $\rel$ is
symmetric with respect to any permutation of coordinates,
$\{c\}\tm\zA\tm\{d\}\sse\rel$. In particular, $(c,c,d)\in\rel$. Then
$$
h\left(\cll aab,\cll ccd\right)=\cll bba,
$$
as $c=\vf^{-1}(d)$ and $a=\vf^{-1}(b)$. The tuple $(b,b,a)$ is as required and we can set $c'=d'=b$. 
The case when $\Sg{c, b} = \zA$ is dual by applying $\vf$, $c$ becomes $d$ 
and we get $\Sg{a, d} = \zA$. Then the same argument proceeds with 
$(c, d, d)\in R$ and $h(a, d) = b$.

Thus, $(a,a,a)\in\rel$ which means that there is a term operation
$f(x,y,z)$ such that $f(a,a,b)=f(a,b,a)=f(b,a,a)=a$. Since $\vf$ is an
automorphism, we also get $f(b,b,a)=f(b,a,b)=f(a,b,b)=b$, i.e.\ $f$ is
a majority operation on $\{a,b\}$.
\end{proof}

% general theorem
\begin{proof}[Proof of Theorem~\ref{the:connectedness}]
Suppose for contradiction that $\cG(\zA)$ is disconnected. Let
$\zB$ be a minimal subalgebra of $\zA$ such that $\cG(\zB)$ is
disconnected.
Let $\th$ be a maximal congruence of $\zB$.
As is easily seen, if $\cG(\zB\fac\th)$ is connected then $\cG(\zB)$ is
connected. Indeed, if $c\fac\th d\fac\th$ is an edge in $\zB\fac\th$ then $cd$ is an edge in $\zB$ (see also Lemma~\ref{lem:edge-factor} below). Also, $\th$-blocks are subalgebras of $\zB$ and therefore connected by the choice of $\zB$.
Take $c,d\in\zB$; let $c'=c\fac\th, d'=d\fac\th$. If $\Sg{c,d}\ne\zB$
then $c,d$ are connected in $\Sg{c,d}$ by the assumption made.
Otherwise $\Sg{c',d'}=\zB\fac\th$ and we are in the conditions of
Proposition~\ref{pro:simple}. If $\zB\fac\th$ is a set, $cd$ is an edge
of the unary type. If $\zB\fac\th$ is term equivalent to the full
idempotent reduct of a module, then $cd$ is an edge of the affine
type. If $c',d'$ are connected in $\cH(\zB\fac\th)$ as in items (3),(4)
of Proposition~\ref{pro:simple},
then $c,d$ are connected in $\cG(\zB)$ by the assumption of the
minimality of $\zB$. In the remaining cases $cd$ is a semilattice or
majority edge.

Finally, if $\cG(\zA)$ has an edge $ab$ of the unary type and $\th$ is a 
maximal congruence of $\Sg{a,b}$ certifying that, then
$\{a\fac\th,b\fac\th\}$ is a factor of $\zA$ term equivalent
to a set, and so $\H\S(\zA)$ admits type \one. Similarly, if $\cG(\zA)$
has an affine edge then $\H\S(\zA)$ admits type \two.
Conversely, suppose $\H\S(\zA)$ admits type \one. Then some algebra 
$\zB$ from $\H\S(\zA)$ has a prime interval $\al\prec\beta$ in
$\Con(\zB)$ of type \one. Taking the factor of $\zB$ modulo $\al$ and
restricting $\zB$ on a nontrivial $\beta$-block, $\zB$ can be assumed 
to be simple and $\zz_\zB\prec\zo_\zB$ has the unary type. Therefore, by 
Proposition~\ref{pro:simple-abelian} $\zB$ is a 2-element set. Let $a',b'$ be 
the elements of $\zB$. Then as $\zB\in\H\S(\zA)$ there are elements 
$a,b\in\zA$ such that $a\in a',b\in b'$. As is easily seen, $ab$ is an 
edge of the unary type. Similarly, if some factor of $\zA$ has a prime
congruence interval of type \two, we can find an edge of the affine type
in $\cG(\zA)$.
\end{proof}

We complete this section observing several simple properties of edges 
in connection with subalgebras and factor algebras. 

\begin{lemma}\label{lem:edge-factor}
Let $\zA$ be an idempotent algebra and $\al$ a congruence of $\zA$.
Then\\[2mm] 
(1) If $a\fac\al b\fac\al$, $a,b\in\zA$, is an edge in $\cG(\zA\fac\al)$ then 
$ab$ is also an edge in $\cG(\zA)$ of the same type.\\[2mm]
(2) If $\cH(\zA\fac\al)$ is connected then so is $\cH(\zA)$.
\end{lemma}

\begin{proof}
(1) Let $\th$ be a congruence of $\zB=\Sgg{\zA\fac\al}{a\fac\al,b\fac\al}$ 
witnessing the type of the edge $a\fac\al b\fac\al$. Let $\th'\ge\al$ be the congruence of $\zB'=\{a\in\zA\mid a\fac\al\in\zB\}$ such that $\th'\fac\al=\th$, and consider the congruence
$\eta$ which is the restriction of $\th'$ to the algebra $\Sgg\zA{a,b}$.
Clearly, $\Sgg\zA{a,b}\fac\eta$ is isomorphic to $\zB\fac\th$, thus
witnessing that $ab$ is an edge in $\cG(\zA)$ of the same type as 
$a\fac\al b\fac\al$.

(2) This statement follows from the fact that any $a,b\in\zA$ such that 
$(a,b)\in\al$ are connected by the subalgebra $a\fac\al$.
\end{proof}

The following example shows that the converse of 
Lemma~\ref{lem:edge-factor}(1) does not always hold.

\begin{example}\label{exa:no-edge-factor}
Let $\zA$ be an algebra containing a pair that is not an edge. 
For example, we may consider the algebra from 
Example~\ref{exa:no-edge}, that is, the universe of $\zA$ is
$A=\{a,b,c\}$ and the only basic operations are the binary 
operations $f,g$ acting as described in Example~\ref{exa:no-edge}. 
For a non-edge
consider the pair $ab$. Let $\zB=(B,m)$, where $B=\{0,1\}$ and 
$m$ is the majority operation on $\{0,1\}$. We define $f,g$ on $B$
to be the first projections and $m$ on $A$ to be the first projection as well.
Let $\zA'=(A,f,g,m)$, $\zB'=(B,f,g,m)$, and consider $\zC=\zA'\tm\zB'$.
Let also $\pi_1,\pi_2$ be the projection congruences of $\zC$.
It is easy to see that $(a,0)(b,1)$ is an edge of the majority type
as witnessed by the congruence $\pi_2$. Indeed, $(a,0),(b,1)$
generate $\zC$ and $\zC\fac{\pi_2}$ is isomorphic to $\zB'$,
which is an edge of the majority type. On the other hand,
$\zC\fac{\pi_1}$ is isomorphic to $\zA'$, where 
$(a,0)\fac{\pi_1}, (b,1)\fac{\pi_1}$ correspond to $a,b$, respectively.
Thus, $(a,0)\fac{\pi_1}(b,1)\fac{\pi_1}$ is not an edge.
\end{example}

The next statement follows straightforwardly from definitions.

\begin{lemma}\label{lem:edge-subalgebra}
Let $\zA$ be an idempotent algebra and $\zB$ its subalgebra. Then 
for any $a,b\in\zB$\\[2mm]
(1) the pair $ab$ is an edge in $\cG(\zB)$ if and only 
if it is an edge in $\cG(\zA)$, and it has the same type in both algebras;\\[2mm]
(2) if elements $a,b$ are connected in $\cH(\zB)$ then they are connected in 
$\cH(\zA)$.
\end{lemma}

\begin{lemma}\label{lem:many-edges}
Let $\zA$ be an idempotent algebra and $ab$ an edge of $\cG(\zA)$.
Let $\th$ be a congruence of $\Sg{a,b}$ witnessing that $ab$ is
an edge. Then for any $c\in a\fac\th$, $d\in b\fac\th$, the pair
$cd$ is also an edge of $\cG(\zA)$ of the same type.
\end{lemma}

\begin{proof}
Consider the algebra $\zB=\Sg{a,b}\fac\th$. The fact that $ab$
is an edge and its type only depends on the term operations of $\zB$.
Since $c\in a\fac\th$, $d\in b\fac\th$, the result follows.
\end{proof}

%%%%%%%%%%%%%%%%%%%%%%%%%%%%%%%%%%
\subsection{Adding thick edges}\label{sec:adding-thick}

Generally, an edge, or even a thick edge is not a subalgebra. However, it
is always possible to find a reduct for which every (thick) edge is a 
subalgebra. For instance, one can throw away all the term operations of an 
algebra. Every subset of such a reduct is a subalgebra. More difficult is to 
find a reduct that keeps the types of edges in $\cG(\zA)$. In this section
we show that every idempotent algebra $\zA$ omitting type \one\ has
a reduct $\zA'$ such that every one of its (thick) edges of the semilattice or majority
type is a subalgebra of $\zA'$, and if $\zA$ omits certain types then so
does $\zA'$. More precisely, we prove the following

\begin{theorem}\label{the:adding}
Let $\zA$ be an idempotent algebra. There exists a reduct $\zA'$ of
$\zA$ such that every thick edge of the semilattice or majority type is a subuniverse of $\zA'$ and\\[2mm]
(1) if 
$\cG(\zA)$ does not contain edges of the unary type, then $\cG(\zA')$ 
does not contain edges of the unary type;\\[2mm]
(2) if $\cG(\zA)$ contains no edges of the unary and affine types, 
then $\cG(\zA')$ contains no edges of the unary and affine types.
\end{theorem}

\begin{remark}\label{rem:taylor-minimal}\rm
If we do not insist on a specific reduct in item (1) of Theorem~\ref{the:adding}, 
it allows for a  simpler proof
than we give here. It was suggested by Brady and works for
\emph{Taylor-minimal} algebras introduced in 
\cite{Brady17:taylor-minimal}.
A Taylor-minimal algebra is a finite Taylor algebra, that is, it has a
term operation satisfying some Taylor identities, and whose clone of term
operations is minimal among clones containing a Taylor operation. First,
it can be shown that if $\zA$ is an idempotent algebra such that
$\H\S(\zA)$ omits type one, there is a reduct of $\zA$ that is a Taylor-minimal
algebra. It suffices therefore to show that every thick edge of a Taylor-minimal
algebra is a subalgebra. Suppose that $\zA$ is a Taylor-minimal algebra,
then it has a cyclic term $f$, which can be assumed to be the only basic
operation of $\zA$. Let also $ab$ be a semilattice or majority edge
of $\cG(\zA)$ and congruence $\th$ witnesses that. Then the semilattice or
majority operation on $\{a\fac\th,b\fac\th\}$ is a cyclic term on this set, 
denote it $t$. The idea now is to compose $f$ and $t$ in such a way 
that the resulting term is still cyclic on $\zA$, but preserves 
$a\fac\th\cup b\fac\th$. Since $\zA$ is Taylor-minimal, this implies 
that $a\fac\th\cup b\fac\th$ is a subalgebra of $\zA$.

It is not clear, however, how to extend the argument above to item
(2). Also, item (1) can be proved for a wider class of algebras than 
Taylor-minimal ones. So, here we give a longer, but more general proof.
\end{remark}

\begin{remark}\label{rem:smooth}\rm
If an idempotent algebra satisfies the property that every one of its (thick) 
semilattice or majority edges is a subalgebra, we call it \emph{smooth}.
\end{remark}

The main auxiliary statement to prove Theorem~\ref{the:adding} is
Proposition~\ref{pro:adding}. We say that an algebra $\zA$ 
is \emph{$X$-connected}, where 
$X\sse\{$unary, semilattice, ma\-jo\-ri\-ty, affine$\}$, if for any 
subalgebra $\zB$ of $\zA$ and any $a,b\in\zB$ there is a path
in $\cG(\zB)$ connecting $a$ and $b$ and that only contains
edges of types from $X$. 
\begin{prop}\label{pro:adding}
Let $\zA$ be an idempotent algebra, let
$ab$ be an edge of $\cG(\zA)$ of semilattice or majority type and
$\th$ a congruence witnessing that, and let
$\rel_{ab}=a\fac\th\cup b\fac\th$.
Let also $F_{ab}$ denote the set of all term operations of $\zA$ preserving
$\rel_{ab}$ and $\zA'=(A,F_{ab})$. Then\\[2mm]
(1) if $\zA$ is $\{\text{semilattice, majority, affine}\}$-connected
then so is $\zA'$;\\[2mm]
(2) if $\zA$ is $\{\text{semilattice,majority}\}$-connected, then so is 
$\zA'$. 
\end{prop}

First, we show how Theorem~\ref{the:adding} is obtained from 
Proposition~\ref{pro:adding}.

\begin{proof}[Proof of Theorem~\ref{the:adding}]
Let $X$ be one of the sets used in Theorem~\ref{the:adding}, that is, $X$ is $\{$semilattice,majority$\}$, or $\{$semilattice,majority,affine$\}$. We
prove that if $ab$ is an edge of $\cG(\zB)$  of the unary or affine type for some algebra $\zB$, then there are $c,d\in\Sg{a,b}$ such that any path in $\cG(\Sg{c,d})$ from $c$ to $d$ contains an edge of the unary or affine type, respectively. Hence, if $\cG(\zA')$ contains an edge whose type is not in $X$, it is not $X$-connected.
By Proposition~\ref{pro:adding}, neither is $\zA$, which means that if $\zA'$ has a unary edge, then $\zA$ has a unary edge, while if $\zA'$ has no unary edge, but has an affine edge, the same property holds for $\zA$. Proceeding inductively, we reach a reduct in which all thick edges of the semilattice and majority type are subuniverses, which has a unary edge iff $\zA$ has a unary edge, and the same reduct has no unary edge nor an affine edge iff $\zA$ has no unary edge nor an affine edge.
 
Let $\th$ be a congruence of $\zB=\Sg{a,b}$ witnessing that $ab$ is an edge. Choose $c,d\in\zB$ such that $c\fac\th\ne d\fac\th$ and $\zC=\Sg{c,d}$ is minimal possible with this condition. Suppose also that $c=c_1,c_2\zd c_k=d$ is a 
path in $\cG(\zC)$ connecting $c$ and $d$. Then for some $i\in[k-1]$
it holds $c_i\fac\th\ne c_{i+1}\fac\th$. We show that $c_ic_{i+1}$
is an edge of $\cG(\zA)$ of the same type as $ab$.

By the choice of $c,d$ we have $\Sg{c_i,c_{i+1}}=\zC$. Let $\eta=\th\red\zC$. Clearly $\zC'=\zC\fac\eta$ is isomorphic to a subalgebra of $\zB'=\zB\fac\th$. As $\zB'$ 
is a set or a module depending on whether $ab$ has the unary or
affine type, so is $\zC'$. Therefore $\eta$ witnesses that $c_ic_{i+1}$
is an edge of the unary or affine type depending on whether $ab$ has 
the unary or affine type. It remains to show that $c_ic_{i+1}$ does not have any other type. Take a proper congruence $\chi$ of $\zC$, note that $\chi$ separates $c_i,c_{i+1}$ as they generate $\zC$. If $\chi\not\le\eta$ then some $\chi$-block $D$ is not contained in an $\eta$-block, a contradiction with the choice of $c,d$, as $D$ is a subuniverse and a proper subset of $\zC$. 
Finally, if $\chi\le\eta$, then a semilattice or majority operation on $\{c_i\fac\chi,c_{i+1}\fac\chi\}$ gives rise to a semilattice or majority operation on $\{c_i\fac\eta,c_{i+1}\fac\eta\}$ contradicting the assumption that $\zC$ is a module or a set.
\end{proof}

\begin{lemma}\label{lem:good-functions}
(1) Let $ab$ be a semilattice edge of an algebra $\zA$, let $\th$ be a
congruence of $\zA$ witnessing that, and $f$ a binary term operation
which is semilattice on $\{a\fac\th,b\fac\th\}$. Then $f$
can be chosen to satisfy (on $\zA$) any one of the two equations:
\[
f(x,f(x,y))=f(x,y), \qquad
f(f(x,y),f(y,x))=f(x,y).
\]
(2) Let $ab$ be a majority edge of an algebra $\zA$, let $\th$ be a
congruence of $\zA$ witnessing that, and $m$ a ternary term operation
which is majority on $\{a\fac\th,b\fac\th\}$. Then $m$
can be chosen to satisfy (on $\zA$) any one of the two equations:
\begin{align*}
m(x,m(x,y,y),m(x,y,y)) &= m(x,y,y), \\
m(m(x,y,z),m(y,z,x),m(z,x,y)) &= m(x,y,z).
\end{align*}
Moreover, if for some $c,d\in\zA$ and a congruence $\eta$ of $\Sgo{c,d}$ there is an operation that is semilattice or majority on $\{a\fac\th,b\fac\th\}$ and the first projection on $\{c\fac\eta,d\fac\eta\}$, then the operations $f,m$ above can be chosen such that they are the first projection on $\{c\fac\eta,d\fac\eta\}$.
\end{lemma}

\begin{proof}
(1) To show that $f$ can be chosen to satisfy the equation
$f(x,f(x,y))=f(x,y)$, for every $x,y\in A$, we consider the unary operation
$g_x(y)=f(x,y)$. There is a natural number $n_x$ such that
$g_x^{n_x}$ is an idempotent transformation of $A$. Let $n$ be
the least common multiple of the $n_x$, $x\in A$, and
$$
h(x,y)=f(\underbrace{x,f(x,\ldots f(x}_{\mbox{\footnotesize $n$
times}},y)\ldots)).
$$
Since $g^n_x(y)$ is idempotent for any $x\in A$, we have
$h(x,h(x,y))=g^n_x(g^n_x(y))=g^n_x(y)=h(x,y)$. Finally, as is easily
seen $h$ equals $f$ on $\{a\fac\th,b\fac\th\}$.

To show that $f$ can be chosen to satisfy the second equation,
consider the unary operation $g$ on $\zA^2$ given by
$(x,y)\mapsto(f(x,y),f(y,x))$. There is $n$ such that $g^n$ is
idempotent. Then
\[
h(x,y)=\underbrace{f(\dots f(f(f(}_{\mbox{\footnotesize $n$
times}}x,y),f(y,x)),f(f(y,x),f(x,y))\dots)
\]
satisfies the required equation and equals $f$ on
$\{a\fac\th,b\fac\th\}$.

\smallskip

(2) To show that $m$ can be chosen to satisfy the equation
$m(x,m(x,y,y),\lb m(x,y,y))=m(x,y,y)$, for every $x\in A$, we consider the
unary operation $g_x(y)=m(x,y,y)$. There is
a natural number $n_x$ such that $g_x^{n_x}$ is an idempotent
transformation of $A$. Let $n$ be the least common multiple of the
$n_x$, $x\in A$, and
\begin{align*}
h(x,y,z)&=g_x^{n-1}(m(x,y,z))=m(\underbrace{x,m(x,\ldots 
m(x}_{\mbox{\footnotesize $n$ times}},y,z),m(x,y,z)\ldots)),
\end{align*}
Observe that $h(x,y,y)=g_x^n(y)$.
Since $g^n_x(y)$ is idempotent for any $x\in A$, we have
$h(x,h(x,y,y),h(x,y,y))=g^n_x(g^n_x(y))=g^n_x(y)=h(x,y,y)$.
Finally, as is easily seen $h$ is a majority operation on
$\{a\fac\th,b\fac\th\}$.

For the second equation consider the unary operation $g$ on $\zA^3$
given by $(x,y,z)\mapsto(m(x,y,z),m(y,z,x),m(z,x,y))$. There is $n$ such
that $g^n$ is idempotent. Then
\begin{align*}
h(x,y,z) &= \underbrace{m(\dots m(m(m(}_{\mbox{\footnotesize $n$
times}}x,y,z),m(y,z,x),m(z,x,y)))\dots).
\end{align*}
satisfies the required equation and equals $m$ on
$\{a\fac\th,b\fac\th\}$.

\smallskip

To prove the last statement of the lemma it suffices to observe that if we start with $f$ or $m$ that is the first projection on $\{c\fac\eta,d\fac\eta\}$, the resulting operation also satisfies this property.
\end{proof}

In the rest of this section we assume the conditions and notation used in Proposition~\ref{pro:adding}.
Recall that the subalgebra of $\zA$ generated by a set
$B\sse A$ is denoted by $\Sgo B$, while the subalgebra of
$\zA'$ generated by the same set is denoted by $\Sgn B$. In
general, $\Sgn B\sse\Sgo B$. In what follows, let $f$
[respectively, $m$] be a term operation of $\zA$ that witnesses
that $ab$ is a semilattice [respectively, majority] edge, that is,
such that $f\fac\th$ [respectively, $m\fac\th$] is a
semilattice [respectively, majority] operation on
$B'=\{a\fac\th,b\fac\th\}$. 
We start with two lemmas.

\begin{lemma}\label{lem:module-projection}
Let $c,d\in\zA$ be such that there exists a congruence $\eta$ of $\zC=\Sgn{c,d}$ such that $\zC'=\Sgn{c,d}\fac\eta$ is a subalgebra of a reduct of a module. Then operation $f$ or $m$ can be chosen to be the first projection on $\zC'$.
\end{lemma}

\begin{proof}
As $\zC'$ is a  reduct of a module over some ring $\zK$, the operations $f$ or $m$ (note, they belong to $F_{ab}$) on $\zC'$ have the form
\[
f(x,y)=\al x+(1-\al)y, \qquad m(x,y,z)=\al x+\beta y+\gm z,\ \ \al+\beta+\gm=1.
\]
Let $n$ be such that $\al^n$ is an idempotent of $\zK$. Then set
\[
f'(x,y)=\underbrace{f(f(\ldots f}_{\mbox{\footnotesize $n$ times}}(x,y) \ldots,y),y).
\]
Since $f$ and $f'$ are idempotent, $f'(x,y)=\al^nx+(1-\al^n)y$ on $\zC'$ and $f'(x,y)=f(x,y)$ on $B'$. If $\al^n\in\{0,1\}$, then $f'(x,y)$ or $f'(y,x)$ is the first projection on $\zC'$ and we are done. If $\al^n$ is a nontrivial idempotent, set
\[
f''(x,y)=f'(f'(x,y),x).
\]
Again, $f''(x,y)=f(x,y)$ on $B'$. On the other hand, on $\zC'$ we have
\[
f''(x,y)=((\al^n)^2+(1-\al^n))x+\al^n(1-\al^n)y=x.
\]
 
For the operation $m$ we need to perform several steps similar to the ones above. The goal is to make sure that $\al,\beta,\gm$ can be assumed to be idempotents of $\zK$ and such that $\al\beta=\al\gm=\beta\gm=0$. For such $m$ we can set 
\begin{align*}
m'(x,y,z) &= m(m(x,y,z),m(z,x,y),m(y,z,x))\\
&= (\al^2+\beta^2+\gm^2)x +(\al\beta+\beta\gm+\gm\al)y+(\al\gm+\beta\al+\gm\beta)z\\
&= x+\gm\al y+(\beta\al+\gm\beta)z.
\end{align*}
As $m'$ is idempotent, $p(x,y)=m'(x,y,y)=x$ on $\zC'$ and $p(x,y)=y$ on $B'$. Then let $m''(x,y,z)=p(x,m(x,y,z))$ which is equal to $m(x,y,z)$ on $B'$, and $m''(x,y,z)=x$ on $\zC'$.

It remains to prove that $m$ with the required properties exists. Again, assume that $\al^n$ is an idempotent of $\zK$ and set 
\[
m'(x,y,z)=\underbrace{m(m(\ldots m}_{\mbox{\footnotesize $n$ times}}(x,y,z) \ldots,y,z),y,z).
\]
We have $m'(x,y,z)=\al^nx+\beta'y+\gm'z$ on $\zC'$ for some $\beta',\gm'\in\zK$ with $\al^n+\beta'+\gm'=1$, and, as is easily seen, $m'$ is a majority operation on $B'$. Next, set 
\begin{align*}
m''(x,y,z) &=m'(x,m'(x,y,z),m'(x,y,z))\\
&=(\al^n+(\beta'+\gm')\al^n) x+(\beta'+\gm')\beta'y+(\beta'+\gm')\gm'z.
\end{align*}
Since $\beta'+\gm'=1-\al^n$ and so $\al^n(\beta'+\gm')=(\beta'+\gm')\al^n=0$, we get
\begin{align*}
\al^n+(\beta'+\gm')\al^n &= \al^n,\\
\al^n(\beta'+\gm')\beta' &= 0,\\
\al^n(\beta'+\gm')\gm' &= 0.
\end{align*}
Since $m''$ is still a majority operation on $B'$, we may assume that for $m$ it holds that $\al^2=\al$, $\al\beta=\al\gm=0$. 

Next, let $\beta^\ell$ be an idempotent of $\zK$. Repeat the steps above for $y$:
\[
m'(x,y,z)=\underbrace{m(x,m(x,\ldots m}_{\mbox{\footnotesize $\ell$ times}}(x,y,z) \ldots,z),z).
\]
As is easily seen, $m'$ is a majority operation on $B'$ and
\[
m'(x,y,z)=(1+\beta+\dots+\beta^{\ell-1})\al x+\beta^\ell y+\gm'z
\]
for some $\gm'\in\zK$. Since $\al\beta=0$, for the first coefficient we have 
\begin{align*}
((1+\beta+\dots+\beta^{\ell-1})\al)^2 &= (1+\beta+\dots+\beta^{\ell-1})\al+\sum_{i=0,j=1}^{\ell-1,\ell-1}\beta^i\al\beta^j\al\\
&= (1+\beta+\dots+\beta^{\ell-1})\al.
\end{align*}
Thus, the first and second coefficients are idempotent. Denote them $\al'$ and $\beta'$, respectively. Note also that $\al'\beta'=0$ and $\al'\gm'=\al'(1-\al'-\beta')=0$. Next, set 
\begin{align*}
m''(x,y,z) &=m'(m'(x,y,z),y,m'(x,y,z))\\
&=(\al'+\gm')\al' x+(\al'+1+\gm')\beta'y+(\al'+\gm')\gm'z.
\end{align*}
We have 
\begin{align*}
& ((\al'+\gm')\al')^2 = (\al'+\gm'\al')^2=\al'+\al'\gm'\al'+\gm'\al'^2+\gm'\al'\gm'\al'=\al'+\gm'\al',\\
& ((\al'+1+\gm')\beta')^2 = (\beta'+(1-\beta')\beta')^2=\beta'^2=\beta',\\
& (\al'+\gm')\al'(\al'+1+\gm')\beta' = (\al'+\gm')\al'\beta'=0,\\
& (\al'+\gm')\al'(\al'+\gm')\gm' = (\al'+\gm')\al'\gm'=0\\
& (\al'+1+\gm')\beta'(\al'+\gm')\gm' = (\al'+1+\gm')\beta'(1-\beta')\gm'=0.
\end{align*}
Moreover, it is straightforward that $m''$ is a majority operation on $B'$. Assuming that the original operation $m$ satisfies these conditions we have $\al^2=\al, \beta^2=\beta, \al\beta=\al\gm=\beta\gm=0$. 

Finally, we repeat the first step again for $z$ by assuming that $\gm^k$ is an idempotent of $\zK$ and setting
\[
m'(x,y,z)=\underbrace{m(x,y,m(x,y,\ldots m}_{\mbox{\footnotesize $k$ times}}(x,y,z) \ldots)).
\]
As before, it is straightforward to verify that $m'$ is a majority operation on $B'$ and that the coefficients $\al',\beta',\gm'$ of $m'$ on $\zC'$ satisfy the conditions $\al'^2=\al', \beta'^2=\beta', \gm'^2=\gm'$, $\al'\beta'=\al'\gm'=\beta'\gm'=0$. 
\end{proof}

Notice that if $\Sgg\zA{c,d}$ is a set or a module, then $\Sgg{\zA'}{c,d}$ is a subalgebra of a reduct of a module. Hence we obtain the following corollary.

\begin{corollary}\label{cor:module-projection}
If $cd$ is a unary or affine edge witnessed by congruence $\eta$
of $\Sgo{c,d}$, then $f$ or $m$ can be chosen to be the first 
projection on $\Sgo{c,d}\fac\eta$.
\end{corollary}

\begin{lemma}\label{lem:subalgebra-old-new}
Let $c,d\in\zA$ be such that $cd$ is an edge (of any type including the
unary type) in $\zA'$. Let also $\eta$ be a congruence of 
$\Sgn{c,d}$ witnessing that, and assume that $\Sgo{c',d'}=\Sgo{c,d}$, 
for any $c',d'\in\Sgn{c,d}$ with $c'\in c\fac\eta,d'\in d\fac\eta$. Then one of the following holds.
\begin{itemize} 
\item[(i)]
$cd$ is a semilattice edge or there is $e\in\Sgn{c,d}$ such that $ce,de$ are semilattice edges
in $\zA'$, or
\item[(ii)]
$\cH(\Sgn{c,d})$ is connected, or
\item[(iii)]
$\Sgn{c,d}=\Sgo{c,d}$ as sets, or
\item[(iv)]
$ab$ is a majority edge in $\zA$ and $cd$ is a semilattice or majority edge in 
$\Sgn{c,d}$.
\end{itemize}
\end{lemma}

\begin{proof}
Let $\zC=\Sgn{c,d}$. 

\medskip
\noindent
{\sc Case 1.}
$ab$ is a semilattice edge.

Recall that $f$ is a term operation of $\zA$ semilattice on 
$\{a\fac\th,b\fac\th\}$, where $\th$ is a congruence of $\Sgo{a,b}$ witnessing
that $ab$ is a semilattice edge. By Lemma~\ref{lem:good-functions}(1), $f$ can be assumed to satisfy the equation $f(x,f(x,y))=f(x,y)$. We  consider the congruence $\chi$ of $\zC$ generated by the set\lb $D=\{(f(c',d'),f(d',c'))\mid c',d'\in\zC\}$. Note that 
if $\chi$ is not the total congruence then $f$ is commutative on $\zC\fac\chi$.

\medskip

\noindent
{\sc Case 1.1.}
$\chi$ is not the total congruence.\\[2mm]
As $f$ satisfies the equation $f(x,f(x,y))=f(x,y)$, it is a semilattice operation on $\{c\fac\chi,f(c\fac\chi,d\fac\chi)\}$ and
$\{f(c\fac\chi,d\fac\chi),d\fac\chi\}$, implying condition (i) of Lemma~\ref{lem:subalgebra-old-new}, where $e=f(c,d)$, unless $f(c\fac\chi,d\fac\chi)\in\{c\fac\chi,d\fac\chi\}$, in which case $cd$ is a semilattice edge.

\medskip

\noindent
{\sc Case 1.2.} 
$\chi$ is the total congruence\\[2mm]
The congruence $\chi$ 
is the transitive closure of the set 
$D'=\{(p(c'),p(d'))\mid (c',d')\in D, \text{ and $p$ is a unary polynomial of
$\zC$}\}$. For every such unary polynomial there is a term operation
$g$ of $\zC$ such that $p(x)=g(c,d,x)$.
We first show that if, for every pair $(g(c,d,c')$, $g(c,d,d'))$, where $(c',d')\in D$,
and $g$ is a term operation of $\zC$,
the subalgebra $\Sgo{g(c,d,c'),g(c,d,d')}$ of $\zA$ is a
proper subalgebra of $\Sgo{c,d}$, then $\cH(\zC)$ is connected implying 
condition (ii) or (ii).
Firstly, since $\chi$ is the total congruence, any two elements from $\zC$ are 
connected by a sequence of subalgebras of $\zC$ generated by 
pairs from $D'$. Secondly, all these subalgebras
are proper. Indeed, if $\Sgn{g(c,d,c'),g(c,d,d')}=\zC$ for some $c',d'$, then
$c,d\in\Sgn{g(c,d,c'),g(c,d,d')}\sse\Sgo{g(c,d,c'),g(c,d,d')}$. Therefore,
$\Sgo{g(c,d,c'),g(c,d,d')}=\Sgo{c,d}$, a contradiction. 

Suppose now that, for a certain $(c',d')\in D$ and a ternary term
operation $g$ of $\zC$, we have $\Sgo{g(c,d,c'),g(c,d,d')}=\Sgo{c,d}$.
Then for any $e\in\Sgo{c,d}$, there is a term operation $h$ of 
$\zA$ such that $h(g(c,d,c'),g(c,d,d'))=e$. Consider 
$h'(x,y,z,t)=h(g(x,y,f(z,t)),g(x,y,f(t,z)))$. We have\lb
$h'\red{B'}(x,y,z,t)=g(x,y,f(z,t))$, where $B'=\{a\fac\th,b\fac\th\}$.
Hence, $h'\in F_{ab}$. On the
other hand, $h'(c,d,c'',d'')=e$, where $c'',d''$ are such that $f(c'',d'')=c'$,
$f(d'',c'')=d'$. Thus, $\Sgn{c,d}=\Sgo{c,d}$, and we obtain item (iii).

\medskip

\noindent
{\sc Case 2.} 
$ab$ is a majority edge.

Let $\chi'$ be the congruence of $\zC$ generated by
\begin{align*}
D&=\{(m(c',d',d'),m(d',c',d')),
(m(c',d',d'),m(d',d',c')),(m(d',d',c'),\\
&  \ \ \ \ \ \ \ \ \ \ \ m(d',c',d'))\mid c',d'\in\zC\}.
\end{align*}
As in Case~1 if $\chi'$ is not the total congruence, then 
$m(x,y,y)=m(y,x,y)=m(y,y,x)$ in $\zC\fac{\chi'}$. We consider two
subcases.

\medskip

\noindent
{\sc Case 2.1.} 
$\chi'$ is not the total congruence.\\[2mm]
Let $\chi$ be a maximal congruence of $\zC$ containing $\chi'$.
If $\eta\not\le\chi$, then, since $\eta\join\chi$ is the total congruence, 
$\cH(\zC)$ is connected by the $\eta$- and $\chi$-blocks implying 
items (ii). So, suppose $\eta\le\chi$. By the assumption $cd$ is an edge and $\eta$ witnesses this. The case when $cd$ has the unary type is impossible, because $m$ is not a
projection on $\zC\fac\chi$, and therefore is not a projection on $\zC\fac\eta$. If $cd$ is of semilattice or majority type, we have case (iv) of the lemma. So, suppose that $cd$ has the affine type, and so $\zC\fac\eta$ is term equivalent to a module. By 
Lemma~\ref{lem:module-projection} there is a term operation 
$\ov m$ such that it is majority on $B'=\{a\fac\th,b\fac\th\}$ and 
$\ov m$ is the first projection on $\zC\fac\eta$. 
As is easily seen, $h(x,y)=\ov m(x,y,y)$ is the second projection on
$B'$ and the first projection on $\zC\fac\eta$. 
We show that $\Sgn{c,d}=\Sgo{c,d}$, thus obtaining item (iii). 
By the assumptions of the lemma $\Sgo{h(c,d),d}=\Sgo{c,d}$, and so it 
suffices to prove that $\Sgn{c,d}=\Sgo{h(c,d),d}$. Let $e\in\Sgo{h(c,d),d}$, 
that is, there is a term operation $g(x,y)$ of $\zA$ such that
$e=g(h(c,d),d)$. The operation $g'(x,y)=g(h(x,y),y)=y$ on $B'$ and
so $g'\in F_{ab}$. On the other hand,
$$
e=g'(c,d)=g(h(c,d),d)\in\Sgn{c,d}.
$$

\medskip

\noindent
{\sc Case 2.2.} 
$\chi'$ is the total congruence of $\zC$.\\[2mm]
Similar to Case~1.2 the congruence generated by $D$ is the transitive 
closure of the set $D'=\{(g(c,d,c'),g(c,d,d'))\mid (c',d')\in D$, 
and $g$ is a term operation of $\zC\}$.
As is shown in Case~1.2, if for every pair $(g(c,d,c'),g(c,d,d'))\in D'$
the subalgebra $\Sgo{g(c,d,c'),g(c,d,d')}$ of $\zA$ is a
proper subalgebra of $\Sgo{c,d}$, then $\cH(\zC)$ is connected, and 
condition (ii) holds.

Suppose that, for a certain $(c',d')\in D$ and a ternary term
operation $g$ of $\zC$, we have
$\Sgo{g(c,d,c'),g(c,d,d')}=\Sgo{c,d}$. Then, for any
$e\in\Sgo{c,d}$, there is a term operation $h$ of $\zA$ such that
$h(g(c,d,c'),g(c,d,d'))=e$ (here we slightly deviate from the argument in
Case~1.2). Without loss of generality we may assume
that $c'=m(c'',d'',d''), d'=m(d'',c'',d'')$ for certain
$c'',d''\in\zC$. Consider
$h'(x,y,z,t)=h(g(x,y,m(z,t,t)),g(x,y,m(t,z,t)))$. We have
$h'\red{B'}(x,y,z,t)=g(x,y,m(z,t,t))=g(x,y,t)$, hence, $h'\in F_{ab}$. On the
other hand, $h'(c,d,c'',d'')=e$. Thus, $\Sgn{c,d}=\Sgo{c,d}$.
\end{proof}

\begin{lemma}\label{lem:old-new-edge}
Let $c,d\in\zA$ and $\zC=\Sgn{c,d}$. Suppose that $cd$ is an edge in $\zA$,
$\eta$ is a congruence of $\Sgo{c,d}$ witnessing that, and for any 
$c'\in c\fac\eta, d'\in d\fac\eta$, it holds $\Sgo{c',d'}=\Sgo{c,d}$.
Then either $c,d$ are connected in $\cH(\zC)$, or $cd$ is a semilattice edge in $\zA'$, or there is $e\in\zC$ such 
that $ce,de$ are semilattice edges of $\zA'$, or 
\begin{itemize}
\item[(1)]
if $cd$ is a majority, affine or unary edge, then $cd$ is an edge of $\zA'$ of the same type as it is in $\zA$.
\item[(2)]
if $cd$ is a semilattice edge in $\zA$, then $cd$ is a semilattice or majority edge  in~$\zA'$.
\end{itemize}
\end{lemma}

\begin{proof}
We consider cases when $ab$ is a semilattice and majority edge separately
proving that either the conclusion of the items (1) or (2) is true, or that $c,d$ are connected in $\cH(\zC)$, or $cd$ is a semilattice edge of $\zA'$, or $ce,de$ are semilattice 
edges of $\zA'$ for some $e\in\zC$. First, we give the part of the argument common for both cases.

Let $f$ or $m$ be a term operation of $\zA$, semilattice or majority on
$B'=\{a\fac\th,b\fac\th\}$, respectively. Let $\eta$ be a congruence of $\Sgo{c,d}$ witnessing that $cd$ is an edge and
$\eta'$ a maximal congruence of $\Sgn{c,d}$ containing 
$\eta\red{\Sgn{c,d}}$; set $\zC'=\Sgn{c,d}\fac{\eta'}$.

By Proposition~\ref{pro:simple} there are four options for $\zC'$:
either $\zC'$ is a set or a module, or $c,d$ are connected in $\cH(\zC')$, or
there is an operation $h$ of $\zC'$ which is either semilattice or
majority on $\{c\fac{\eta'},d\fac{\eta'}\}$.
By Lemmas~\ref{lem:edge-factor}(2) and~\ref{lem:edge-subalgebra}(2) if $c,d$ are connected in $\cH(\zC')$, 
then they are also connected in $\cH(\zC)$ and there is nothing to prove. 
Otherwise $cd$ is an edge in $\zC'$, and hence, by 
Lemmas~\ref{lem:edge-factor}(1) and~\ref{lem:edge-subalgebra}(1)
it is also an edge in $\zA'$. Note that if there are $c'\in c\fac{\eta'}, d'\in d\fac{\eta'}$ such that $\Sgn{c',d'}\subsetneq\zC$, then $c$ and $d$ are connected in $\cH(\zC)$ by the $\eta'$-blocks $c\fac{\eta'},d\fac{\eta'}$ and the subalgebra $\Sgn{c',d'}$. Hence we assume that for any $c'\in c\fac{\eta'}, d'\in d\fac{\eta'}$ it holds that $\Sgn{c',d'}=\zC$. 

Suppose first that $\zC'$ has a term operation that is semilattice or 
majority on $\{c\fac{\eta'},d\fac{\eta'}\}$. Note that $cd$ can only be a semilattice or majority edge of $\zA$. Indeed, if $cd$ is a unary or affine edge of $\zA$, then $\Sgo{c,d}\fac\eta$ is a set or a module, and therefore so are $\zC$ and $\zC'$, implying that this case is impossible. If $cd$ is a semilattice edge of $\zA$, there is nothing to prove. Suppose that $cd$ is a majority
edge of $\zA$, which is witnessed by $\eta$. As we want to show that $cd$ remains a majority edge of $\zA'$, it suffices to show that there is no term operation that is semilattice on $\{c\fac{\eta'},d\fac{\eta'}\}$.  Suppose the contrary, and let $\ell$ be a binary operation that is semilattice on $\{c\fac{\eta'}, d\fac{\eta'}\}$, say, $\ell(c\fac{\eta'}, d\fac{\eta'})=d\fac{\eta'}$. Then $\ell(c,d)=d'$ for some $d'\in d\fac{\eta'}$. By the assumption $\Sgn{c,d'}=\zC$.  
Therefore there is a binary operation $\ell'$ of $\zC$ such that $\ell'(c,d')=d$. Set $\ell''(x,y)=\ell'(x,\ell(x,y))$, for this operation we have $\ell''(c,d)=d$, and $\ell''(d,c)=d''\in d\fac{\eta'}$. By the same argument there exists a binary operation $r$ such that $r(c,d'')=d$. Now consider
$n(x,y)=r(y,\ell''(x,y))$. We have $n(c,d)=r(d,\ell''(c,d))=r(d,d)=d$ and $n(d,c)=r(c,\ell''(d,c))=r(c,d'')=d$. Thus, $cd$ has to be a semilattice edge of $\zA$, a contradiction.

Next, suppose that $\zC'$ is either a set or a module. By
Lemma~\ref{lem:module-projection} $m$ can be chosen to be the first projection on $\zC'$. Also, by 
Lemmas~\ref{lem:good-functions} and~\ref{lem:module-projection} $f$ can be chosen to be the first projection on $\zC'$ and additionally to satisfy the identity $f(f(x,y),f(y,x))=f(x,y)$ on $\zA'$. Then by Lemma~\ref{lem:subalgebra-old-new} 
either $\cH(\zC)$ is connected, or $cd$ is a semilattice edge of $\zA'$, or there is $e\in\zC$ such that $ce,de$ are 
semilattice edges of $\zA'$, or $\Sgn{c,d}=\Sgo{c,d}$. There is nothing 
to prove in the first three cases, so suppose that $\Sgn{c,d}=\Sgo{c,d}$.

\smallskip
{\sc Case 1.} $ab$ is a semilattice edge. 

\smallskip

Let $c'=f(c,d), d'=f(d,c)$, and note that $f(c',d')=c', f(d',c')=d'$ and $c'\eqc{\eta'}c, d'\eqc{\eta'}d$. 
Choose a maximal congruence $\eta''$ of $\Sgo{c',d'}$ with $\eta\le\eta''$. By the assumption above $\zC=\Sgn{c,d}=\Sgo{c,d}$, which means that every subalgebra of $\Sgo{c,d}$ is a subalgebra of $\zC$, and $\eta''$ is a congruence of $\zC$. By Proposition~\ref{pro:simple} applied to $\Sgo{c',d'}\fac{\eta''}$ either $c',d'$, and therefore $c,d$ are connected in $\cH(\Sgo{c',d'})$ and hence in $\cH(\zC)$ or $c'd'$ is an edge in $\zA$ witnessed by $\eta''$. If $\eta''\not\le\eta'$, then $c,d$ are connected in $\cH(\zC)$ by $\eta'$- and $\eta''$-blocks. So, suppose that $c'd'$ is an edge in $\zA$ and $\eta''\le\eta'$. As by the assumptions of Proposition~\ref{pro:adding} $\zA$ has no unary edges, let $g$ be a semilattice, majority, or affine operation witnessing that $c'd'$ 
is an edge of $\zA$. Then the operation $g'(x,y)=g(f(x,y),f(y,x))$ in the first case
and the operation
$$
g'(x,y,z)=g(f(x,f(y,z)),f(y,f(z,x)),f(z,f(x,y)))
$$
in the two latter cases belongs to $F_{ab}$ and is a semilattice, majority or
affine operation on $\{c'\fac{\eta''},d'\fac{\eta''}\}$ (or on $\Sgn{c',d'}\fac{\eta''}$ if $g$ is affine), respectively, and hence on $\{c'\fac{\eta'},d'\fac{\eta'}\}$, as well. This shows that 
$c'd'$ cannot be a semilattice or majority edge of $\zA$ in this case, and
if it is an affine edge of $\zA$, then it is also an affine edge of $\zA'$. It remains to show that the same holds for $cd$. If some operation $h$ witnesses that $cd$ is a semilattice or majority edge, since $\eta\le\eta''$, the operation $h$ also witnesses that $\Sgo{c,d}\fac{\eta''}=\Sgo{c',d'}\fac{\eta''}$ is not a module or a set. Finally, if $cd$ is affine and $\Sgo{c,d}\fac\eta$ is a module, then so is $\Sgo{c',d'}\fac{\eta''}$.

\medskip

{\sc Case 2.} $ab$ is a majority edge.

\smallskip

We start with an auxiliary statement.

\medskip

{\sc Claim.}
For any $c_1,c_2\in c\fac{\eta'}$ and  $d_1,d_2,d_3,d_4\in d\fac{\eta'}$ there is a term operation $p$ of $\zC$ with $p(d_1,c_1,d_3)=p(d_2,d_4,c_2)=d_3$ and $p(c\fac{\eta'},d\fac{\eta'},d\fac{\eta'})=c\fac{\eta'}$ on $\zC'$.

\medskip

Let $\rel$ be the subalgebra of $\zC^2$ generated by $\{(d_1,d_2),(c_1,d_4),(d_3,c_2)\}$. Since $\Sgn{d_1,c_1}=\zC$, there is $d'\in\Sgn{d_2,d_4}\sse d\fac{\eta'}$ such that $(d_3,d')\in\rel$. Again, as $\Sgn{c_2,d'}=\zC$, we have $\{d_3\}\tm\zC\sse\rel$, which means in particular that an operation $p$ satisfying the first two equalities exists. 

To show the last condition we consider two cases. First, if $\zC'$ is a set then $p$ is a projection. Also, as we proved $p(d\fac{\eta'},d\fac{\eta'}, c\fac{\eta'})=p(d\fac{\eta'}, c\fac{\eta'},d\fac{\eta'})=d\fac{\eta'}$, it can only be the first projection. Second, if $\zC'$ is a module over a ring $\zK$ then $p(x,y,z)=\al x+\beta y+\gm z$ for some coefficients satisfying $\al+\beta+\gm=1$, where $1$ is the unity of $\zK$. We have 
\[
\al d\fac{\eta'}+\beta d\fac{\eta'}+\gm c\fac{\eta'}=d\fac{\eta'},\qquad \al d\fac{\eta'}+\beta c\fac{\eta'}+\gm d\fac{\eta'}=d\fac{\eta'}.
\]
These equalities imply $\gm c\fac{\eta'}=\gm d\fac{\eta'}$ and $\beta c\fac{\eta'}=\beta d\fac{\eta'}$. Therefore
\[
p( c\fac{\eta'},d\fac{\eta'},d\fac{\eta'})= \al c\fac{\eta'}+\beta d\fac{\eta'}+\gm d\fac{\eta'} = \al c\fac{\eta'}+\beta c\fac{\eta'}+\gm c\fac{\eta'}=c\fac{\eta'}.
\]
The Claim is proved.

\medskip

Let $c_1=m(c,d,d), d_2=m(d,d,c), d_3=m(d,c,d)$, where $c_1\in c\fac{\eta'}$ and $d_2,d_3\in d\fac{\eta'}$. Using the Claim we show that $m$ can be chosen such that $d_2=d_3$. There is a term operation $p$ of $\zC$ such that $p(d_2,d_3,c_1)=p(d_3,c_1,d_2)=d_2$ and $p(c_1,d_2,d_3)\eqc{\eta'}c$. Then for
\[
m'(x,y,z)=p(m(x,y,z),m(y,z,x),m(z,x,y))
\]
we obtain
\begin{align*}
& m'(c,d,d)=p(c_1,d_2,d_3)\eqc{\eta'}c,\quad m'(d,d,c)=p(d_2,d_3,c_1)=d_2,\\
& m'(d,c,d)=p(d_3,c_1,d_2)=d_2.
\end{align*}
Moreover, as is easily seen, $m'\in F_{ab}$. Thus, we assume $c_1=m(c,d,d), d_2=m(d,d,c)=m(d,c,d)$. Note that at this point we can no longer assume that $m$ is the first projection on $\zC'$, although it still satisfies the conditions $m(c,d,d)\eqc{\eta'}c,m(d,c,d)\eqc{\eta'}m(d,d,c)\eqc{\eta'}d$.

By our assumption $\Sgn{c_1,d_2}=\Sgn{c,d}$. Choose a maximal congruence $\eta''$ of $\Sgo{c,d}$ with $\eta\le\eta''$. By Proposition~\ref{pro:simple} applied to $\Sgo{c,d}\fac{\eta''}$ either $c_1,d_2$, and therefore $c,d$,  are connected in $\cH(\Sgo{c,d})$ and therefore in $\cH(\zC)$, or $c_1d_2$ is an edge in $\zA$ witnessed by $\eta''$. If $\eta''\not\le\eta'$, then $c,d$ are connected in $\cH(\zC)$ by $\eta'$- and $\eta''$-blocks. So, suppose that $c_1d_2$ is an edge in $\zA$ and $\eta''\le\eta'$. 

Let $g$ be a semilattice, majority, or affine operation witnessing that $c_1d_2$ is an edge of $\zA$. Then the operation 
\[
g'(x,y,z)=g(m(x,y,z),m(y,z,x))
\] 
in the first case and the operation
\[
g'(x,y,z)=g(m(x,y,z),m(y,z,x),m(z,x,y))
\]
in the two latter cases belongs to $F_{ab}$ and satisfies the following conditions. If $g$ is a semilattice operation and $g(c_1,d_2)\eqc{\eta''}d_2$, then
\begin{align*}
g'(c,d,d) &=g(m(c,d,d),m(d,d,c)) = g(c_1,d_2)\eqc{\eta''} d_2,\\
g'(d,d,c) &\eqc{\eta''} g'(d,c,d)\eqc{\eta''}d_2.
\end{align*}
If $g(c_1,d_2)\eqc{\eta''}c_1$ then
\begin{align*}
g'(c,d,d) &= g(m(c,d,d),m(d,d,c)) = g(c_1,d_2)\eqc{\eta''} c_1,\\
g'(d,d,c) &=  g(d_2,d_2)\eqc{\eta''} d_2, \\
g'(d,c,d) &= g(d_2,c_1)\eqc{\eta''} c_1.
\end{align*}
If $g$ is a majority or affine operation, then 
\begin{align*}
g'(c,d,d) &= g(m(c,d,d),m(d,d,c),m(d,c,d)) = g(c_1,d_2,d_2)\eqc{\eta''} d_2,\\
g'(d,d,c) &\eqc{\eta''} g'(d,c,d)\eqc{\eta''}d_2,
\end{align*}
and
\begin{align*}
g'(c,d,d) &= g(m(c,d,d),m(d,d,c),m(d,c,d)) = g(c_1,d_2,d_2)\eqc{\eta''} c_1,\\
g'(d,d,c) &\eqc{\eta''} c_1,
\end{align*}
respectively. Since $\eta''\le\eta'$, the same equalities hold modulo $\eta'$. If $g$ is semilattice or majority, the equalities above show that $g'$ is not a projection or affine operation on $\zC'$. Thus, $c_1d_2$ cannot be a semilattice or majority edge of $\zA$ in this case, and
if it is an affine edge of $\zA$, then it is also an affine edge of $\zA'$. Then we complete the proof in the same way as in Case~1 using $c_1,d_2$ instead of $c',d'$.
Case~2 and the proof of the lemma are completed.
\end{proof}

\begin{proof}[Proof of Proposition~\ref{pro:adding}.]
We prove items (1) and (2) simultaneously. We need to show that any $c,d\in A$
are connected by a path containing only edges of the correct types. 
We proceed by induction on the cardinality of subuniverses of $\zA'$ to show that for any $\zB\in\Sub(\zA')$, any $c,d\in\zB$ are connected.
For the base case of induction we use $|\zB|=1$.

Let $\zB\in\Sub(\zA')$ and let the statement be proved for all $\zB'\in\Sub(\zA')$ with $|\zB'|<|\zB|$. Observe that if all subalgebras of $\zB$ generated by 2 elements are connected, so is $\zB$. Therefore, it suffices to assume that $\zB$ is generated by two elements,
say, $c,d\in A$. Moreover, by Theorem~\ref{the:connectedness} 
we can assume that $cd$ is an edge of $\zA'$; let congruence 
$\eta$ of $\zC=\Sgn{c,d}$ witness that. Assume that all the proper 
subalgebras of $\zC$ satisfy the $X$-connectedness properties 
given in the proposition. We may also assume that for any 
$c'\in c\fac\eta, d'\in d\fac\eta$, it holds $\Sgo{c',d'}=\Sgo{c,d}$. 
Indeed, if $\Sgo{c',d'}\ne\Sgo{c,d}$, then choose $c',d'$ so that 
$\Sgo{c',d'}$ is minimal possible. If $\Sgn{c',d'}\ne\Sgn{c,d}$,
then by Lemma~\ref{lem:many-edges} $c'd'$ is an edge in $\zA'$ of 
the same type as $cd$, and by the induction
hypothesis satisfies the required conditions. Hence
$c$ and $d$ are connected by $\Sgn{c',d'}$ and the congruence blocks
$c\fac\eta$, $d\fac\eta$. Otherwise as is easily seen, $c'd'$ is an edge 
of $\zA'$ of the same type as $cd$. Therefore, we can replace $c,d$ with 
$c',d'$.

Let $\th$, and $f$ or $m$ be the congruence of $\Sgo{a,b}$, and 
a semilattice or majority operation witnessing the type of $ab$.
We now consider  the 
options given in Lemma~\ref{lem:subalgebra-old-new}. In case (i) elements $c$ and $d$ are connected by a semilattice path. In case 
(ii)  $\cH(\zC)$ is connected and we use the 
inductive hypothesis. In case (iv) there is nothing to prove, because 
$cd$ is an edge of one of the required types. Therefore, assume 
$\Sgn{c,d}=\Sgo{c,d}$. 

The elements $c$ and $d$ are connected by a path
$c=e_1,e_2\zd e_k=d$ in $\cG(\Sgo{c,d})$, where $e_ie_{i+1}$ 
is an edge of $\zA$. Therefore we need to show connectedness for every pair $e_ie_{i+1}$ which is an edge in $\zA$. Observe that the conditions of Lemma~\ref{lem:old-new-edge} can be assumed. Indeed, if $\eta$ is the congruence of $\Sgo{e_i,e_{i+1}}$ witnessing that $e_ie_{i+1}$ is an edge, and $c'\in e_i\fac\eta,d'\in e_{i+1}\fac\eta$ such that $\Sgo{c',d'}\subsetneq\Sgo{e_i,e_{i+1}}$, then $e_i,e_{i+1}$ are connected in $\cH(\Sgn{c,d})$ and the result follows by the induction hypothesis. Now, by Lemma~\ref{lem:old-new-edge} $e_i,e_{i+1}$ are either connected by semilattice edges in $\cG(\Sgn{c,d})$, or they are connected in $\cH(\Sgn{c,d})$, or $e_ie_{i+1}$ is an edge in $\zA'$ of the same type as $e_ie_{i+1}$ if it is unary, affine, or majority in $\zA$, and $e_ie_{i+1}$ is semilattice or majority, if it is semilattice in $\zA$. Therefore all the connectedness conditions are preserved. 
\end{proof}

%%%%%%%%%%%%%%%%%%%%%%%%%%%%%%%%%%%%
\subsection{Unified operations}\label{sec:unified}

% opertion unification
To conclude this section we prove that the term operations
certifying the types of edges of smooth algebras can be significantly unified (cf.\
Proposition~2 from \cite{Bulatov03:conservative}).
\begin{theorem}\label{the:uniform}
Let $\zA$ be an idempotent algebra and $\vc Es$ its thick edges such 
that every $E_i$ of the semilattice or majority type is a subuniverse of $\zA$. 
There are term operations $f,g,h$ of $\zA$ such that for any $E_i=\{a\fac\th,b\fac\th\}$, 
$i\in[s]$, where $\th$ is a congruence witnessing that $ab$ is an edge
\begin{itemize}
\item[(i)]
$f\red{E_i}$ is a semilattice operation if $ab$ is a semilattice edge; 
it is the first projection if $ab$ is a majority or affine edge;
\item[(ii)]
$g\red{E_i}$ is a majority operation if $ab$ is a majority edge;
it is the first projection if $ab$ is an affine edge, and
$g\red{E_i}(x,y,z)=f\red{E_i}(x,f\red{E_i}(y,z))$ if $ab$ is
semilattice;
\item[(iii)]
$h\red{\Sg{a,b}\fac\th}$ is an affine operation if $ab$ is an affine edge; 
it is the first projection if $ab$ is a majority edge, and
$h\red{E_i}(x,y,z)=f\red{E_i}(x,f\red{E_i}(y,z))$ if $ab$ is semilattice.
\end{itemize}
Operations $f,g,h$ are projections on every edge of the unary type.
\end{theorem}

\begin{proof}
First, we show that there is an operation $f$ that is semilattice on
each semilattice edge from $\vc Es$. Let $\vc Bn$ be a list of all 
thick semilattice edges from $\vc Es$. That is, each $B_i$ is a 2-element set.
Let also $\vc fn$ be a list of term operations of the algebra such that
$f_i\red{B_i}$ is a semilattice operation. Notice that every binary
idempotent operation on a 2-element set is either a projection or a
semilattice operation, and every binary operation of a module can be
represented in the form $\al x+(1-\al)y$. 
Since each $f_i$ is idempotent,
for any $i,j$, $f_i\red{B_j}$ is either a projection, or a semilattice
operation. We prove by induction, that the operation $f^i$ constructed
via the following rules is a semilattice operation on $\vc Bi$:
\begin{itemize}
\item
$f^1=f_1$;
\item
$f^i(x,y)=f_i(f^{i-1}(x,y),f^{i-1}(y,x))$.
\end{itemize}

The base case of induction, $i=1$ holds by the choice of
$f_1$. Since $f^{i-1}(x,y)\red{B_j}=f^{i-1}(y,x)\red{B_j}$, for $j\in[i-1]$,
we have 
\[
f^i(x,y)\red{B_j}=f^{i-1}(x,y)\red{B_j}.
\]
Suppose that $f^{i-1}$ satisfies the required conditions. If
$f^{i-1}\red{B_i}$ is a projection, say, $f^{i-1}\red{B_i}(x,y)=x$,
then
$$
f^i(x,y)\red{B_i}=f_i(f^{i-1}(x,y),f^{i-1}(y,x))\red{B_i}=f_i(x,y)\red{B_i},
$$
that is, a semilattice operation on $B_i$. Otherwise 
$f^{i-1}(x,y)=f^{i-1}(y,x)$ on $B_i$, and we have 
$f^i(x,y)\red{B_i}=f^{i-1}(y,x))\red{B_i}$, which is a semilattice operation.

Thus, for each thick edge $E_j$, $j\in[s]$, $f^n\red{E_j}$ is a 
semilattice operation if $E_j$ is semilattice and either a semilattice operation or 
a projection if $E_j$ is majority or unary, and $f^n$ on $\Sg{E_j}$ is $\al x+(1-\al)y$ 
if $E_j$ is affine. However, if $E_j$ is not semilattice, then the subalgebra with the
universe $E_j$ (or $\Sg{E_j}$) has no semilattice operation. Therefore, if $E_j$ is 
majority or unary edge, then $f^n\red{\Sg{E_j}}$ is a projection, and if $E_j$ is an affine edge, then $f^n\red{\Sg{E_j}}$ can be one of the two types, either a projection, or $\al x+(1-\al)y$ for some module $\zM$ over a ring $\zK$ and some $\al\in\zK$. 

Let $\vc D\ell$ be a list of all thick affine edges. Set $f^\dg_0=f^n$ and inductively define $\vc{f^\dg}\ell$ as follows. For $i\in[\ell]$ if $f^\dg_{i-1}\red{\Sg{D_i}}$ is a projection, set $f^\dg_i=f^\dg_{i-1}$. If $f^\dg_{i-1}\red{\Sg{D_i}}=\al x+(1-\al)y$ for some $\zM,\zK$, and $\al\in\zK$, then by Lemma~\ref{lem:module-projection} applied to the algebra $\zA'=(A,f^\dg_{i-1})$ one can transform $f^\dg_{i-1}$ into a projection on $\Sg{D_i}$ without changing its action 
on semilattice and majority edges. Let $f^\dg_i$ be this transformed operation. As is easily seen, $f^\dg_i$ is a projection on $\Sg{D_j}$ for each $j<i$. Therefore $f^n$ can be assumed to be
a projection on every non-semilattice edge $E_j$. Finally,
it is easy to check that $f(x,y)=f^n(f^n(x,y),x)$ satisfies the conditions of the
theorem, regardless on whether $f^n\red{\Sg{E_j}}$ is the first or second 
projection.

Now let $\vc Ck$, $\vc D\ell$ be lists of all thick majority and all thick affine
edges respectively, from $\vc Es$, and $\vc gk$, $\vc h\ell$ lists of
term operations of the algebra $\zA$ such that $g_i\red{C_i}$ is a
majority operation, and $h_i\red{D'_i}$ is the affine
operation, where $D'_i=\Sg{D_i}$. Let $F$ be the set of term operations 
of $\zA$. Notice first, that since neither
$\zC_i=(C_i;F\red{C_i})$ nor $\zD'_i=(D'_i;F\red{D'_i})$ has a
term semilattice operation, every one of their binary term operation is
either a projection or an operation of the form $\al x+(1-\al)y$.
Therefore, for any $i,j$,
$g_i\red{C_j}(x,y,y), g_i\red{C_j}(y,x,y), g_i\red{C_j}(y,y,x)$,
$h_i\red{C_j}(x,y,y), h_i\red{C_j}(y,x,y)$,
$h_i\red{C_j}(y,y,x)\in\{x,y\}$, and also
$g_i\red{D'_j}(x,y,y), g_i\red{D'_j}(y,x,y)$, $g_i\red{D'_j}(y,y,x)$,\lb%
$h_i\red{D'_j}(x,y,y)$, $h_i\red{D'_j}(y,x,y),h_i\red{D'_j}(y,y,x)
\in\{x,y, \al x+(1-\al)y\}$. This means in particular that each of the operations 
$g_i\red{C_j}, h_i\red{C_j}$ is of one of the
following types: a projection, the minority operation, the
majority operation, or a 2/3-minority operation.

First we show that for any $1\le i\le\ell$ there is $h^i$ such that
$h^i\red{D'_j}$ is an affine operation for $j\le i$. As before,
$h^1=h_1$ gives the base case of induction. If $h^{i-1}$ is constructed,
then if $h^{i-1}\red{D'_i}$ is an affine operation then set
$h^i=h^{i-1}$. Otherwise, $h^{i-1}\red{D'_i}=\al x+\beta y+\gm z$
with $\al+\beta+\gm=1$. Then set $h'(x,y)=h^{i-1}(x,y,y)$ and 
observe that $h'(x,y)\red{D'_i}=\al x+(1-\al)y$ and 
$h'(x,y)\red{D'_j}=x$ for $j<i$. Then let
\[
h''_1(x,y,z)=h'(x,h_i(x,y,z)).
\]
Then 
\begin{align*}
h''_1(x,y,z)\red{D'_i} &= \al x+(1-\al)z-(1-\al)y+(1-\al)x \\
& =  x-(1-\al)y+(1-\al)z,\quad \text{and}\\
h''_1(x,y,z)\red{D'_j} &= x, \qquad \text{for $j<i$}.
\end{align*}
Similarly, we can obtain $h''_3(x,y,z)$ with the property
that $h''_3(x,y,z)\red{D'_i}=x-(1-\gm)y+(1-\gm)z$ and 
$h''_3(x,y,z)\red{D'_j}=x$. Furthermore, set
\[
h''_2(x,y,z)=h''_1(h^{i-1}(x,y,z),z,x).
\]
As is easily seen, for this operation we have
\begin{align*}
h''_2(x,y,z)\red{D'_i} &= \al x+\beta y+\gm z-(1-\al)z+(1-\al)x =
x+\beta y-\beta z,\quad \text{and}\\
h''_2(x,y,z)\red{D'_j} &= x-y+z, \qquad \text{for $j<i$}.
\end{align*}
Finally, we set
\[
h^i(x,y,z)=h''_1(h''_2(h''_3(x,y,z),y,z),y,z).
\]
Again, we have 
\begin{align*}
h^i(x,y,z)\red{D'_i} &= x-(1-\gm)y+(1-\gm)z+\beta y-\beta z
-(1-\al)y+(1-\al)z\\
& =  x-y+z,\quad \text{and}\\
h^i(x,y,z)\red{D'_j} &= x-y+z, \qquad \text{for $j<i$}.
\end{align*}

Next, we prove that every $g_i$, $i\in[k]$, can be chosen such that 
$g_i(x,y,z)=x$ on $D'_j$ for all $j\in[\ell]$. 
As in the proof of Lemma~\ref{lem:module-projection}, 
it suffices to find a term operation $p(x,y)$ such that $p(x,y)=x$ on $C_i$ and 
$p(x,y)=y$ on $D'_j$ for $j\in[\ell]$. As $C_i$ is a majority but not a semilattice 
edge, the operation $h^\ell$ from the previous paragraph on $C_i$ can 
be one of the following operations:
a projection, a majority operation, a 2/3-minority operation, or a minority 
operation. If $h^\ell$ is a projection on $C_i$, say, $h^\ell(x,y,z)=x$,
or majority, then set $p(x,y)=h^\ell(x,x,y)$. If $h^\ell$ is 2/3 minority 
satisfying $h^\ell(x,y,y)=y$ 
or $h^\ell(y,y,x)=y$, then set $p(x,y)=h^\ell(y,x,x)$ in the former case, and 
$p(x,y)=h^\ell(x,x,y)$ in the latter case. If $h^\ell$ is 2/3-minority satisfying 
$h^\ell(x,y,x)=x$, then set $p(x,y)=h^\ell(x,h^\ell(x,y,x),x)$; again it is easy 
to check that $p$ satisfies the required conditions.
Finally, if $h^\ell$ on $C_i$ is the minority operation, suppose 
$g_i(x,y,z)=\al_jx+\beta_jy+\gm_jz$ on $D'_j$, $j\in[\ell]$. Then set 
$s_1(x,y)=h^\ell(g_i(x,y,y),y,g_i(y,y,x))$ and $s_2(x,y)=g_i(x,y,x)$. As is easily seen,
$s_1(x,y)=y,s_2(x,y)=x$ on $C_i$ and $s_1(x,y)=s_2(x,y)=(1-\beta_j)x+\beta_j y$
on $D'_j$. Then set $p(x,y)=h^\ell(s_1(x,y),s_2(x,y),y)$. We have 
$p(x,y)=h^\ell(y,x,y)=x$ on $C_i$ and $p(x,y)=y$  on each $D'_j$, $j\in[\ell]$, 
as required.

Now, we prove by induction that for every $1\le i\le k$ there is an
operation $g^i(x,y,z)$ which is majority on $C_j$ for $j\le i$ and is 
the first projection on $D'_r$, $r\in[\ell]$. The
operation $g^1=g_1$ gives the base case of induction. Let us assume
that $g^{i-1}$ is already found. If $g^{i-1}\red{C_j}$ is the majority
operation, set $g^i=g^{i-1}$. Otherwise, it is either a projection, or
a 2/3-minority operation, or the minority operation. In all these case
its variables can be permuted such that
$g^{i-1}\red{C_i}(x,y,y)=x$. Then the operation
$p(x,y)=g^{i-1}(x,y,y)$ satisfies the conditions $p\red{C_i}(x,y)=x$,
and $p\red{C_j}(x,y)=y$ for all $j\in[i-1]$. Therefore, the
operation
$$
g^i(x,y,z)=p(g_i(x,y,z),g^{i-1}(x,y,z))
$$
satisfies the required conditions. Note that $g^i$ is a projection (the same one) 
on every $D'_r$, $r\in[\ell]$, but not necessarily the first projection. So,
we may need to permute the variables of $g^i$ to obtain the desired result.

Operation $g^k$ acts correctly on the majority and affine edges. To make $g^k$ 
act correctly on the semilattice edges we set
$$
g(x,y,z)=g^k(f(x,f(y,z)),f(y,f(z,x)),f(z,f(x,y))).
$$

Finally, set 
\begin{align*}
p(x,y) &= g(x,y,y),\\
\ov h(x,y,z) &= p(h^\ell(x,y,z),x),
\end{align*}
and
$$
h(x,y,z)=\ov h(f(x,f(y,z)),f(y,f(z,x)),f(z,f(x,y))).
$$
Since all the operations are projections on thick edges of the unary type,
as is easily seen $h$ satisfies the conditions required.
\end{proof}

Theorem~\ref{the:uniform} implies that for a smooth algebra $\zA$ there are operations $f,g,h$ that act as the theorem prescribes on all thick edges of $\zA$.  Moreover, it can be easily extended to finite classes of
algebras.

\begin{corollary}\label{cor:uniform}
Let $\cK$ be a finite class of finite smooth algebras. Then there are term
operations $f,g,h$ of $\cK$ such that conditions (i)--(iii) of
Theorem~\ref{the:uniform} are true for any edge $ab$ of any
$\zB\in\cK$.
\end{corollary}

\begin{proof}
Let $\cK=\{\vc\zA n\}$ and $\zA=\tms\zA n$. Observe that for any
$i\in[n]$ and any edge $ab$ of $\zA_i$, any pair of the form
$\ba,\bb\in\zA$, where $\ba[j]=\bb[j]$ for $j\ne i$ and $\ba[i]=a$,
$\bb[i]=b$, is an edge of $\zA$ of the same type as $ab$. We apply 
Theorem~\ref{the:uniform} to the list of thick edges of $\zA$ of the 
form $\{\ba\fac{\ov\th},\bb\fac{\ov\th}\}$ where $\ba,\bb$ are as above for some $i\in[n]$ and edge $ab$ of $\zA_i$, $\th$ is a congruence of $\Sgg{\zA_i}{a,b}$ and $\ov\th$ is the congruence of $\Sgg\zA{\ba,\bb}$ given by $(\bc,\bd)\in\ov\th$ iff $(\bc[i],\bd[i])\in\th$. Then by Theorem~\ref{the:uniform} there are term operations $f,g,h$ of $\zA$, that is, of $\cK$, satisfying the conditions (i)--(iii) of
the theorem on those thick edges. It remains to observe that $f,g,h$ satisfy the conditions (i)--(iii) for each $\zA_i$.
\end{proof}

%%%%%%%%%%%%%%%%%%%%%%%%%%%%%%%%%%%
%%%%%%%%%%%%%%%%%%%%%%%%%%%%%%%%%%%
\section{Thin edges}\label{sec:thin}

Although edges and thick edges as they have been introduced so far reflect 
some aspects of the structure of idempotent algebras, they are not very useful 
from the technical perspective, the way we are going to use them. To improve
the construction of (thick) edges we introduce thin edges, that are always just
a pair of elements, without any congruences or quotient algebras involved. Later
we show that the graph of an algebra $\zA$ based on these thin edges retains 
many of the useful properties of $\cG(\zA)$, most importantly, connectivity.
In a certain sense the connectivity of this new graph is even improved relative to
$\cG(\zA)$. This will come at a price: thin edges are inevitably directed even 
when it does not seem natural or necessary, and are not always subalgebras. Our first goal is to
introduce thin edges, prove their existence, and show that operations $f,g,h$ 
from Theorem~\ref{the:uniform} can be assumed to have a number of 
additional properties.

Fix a finite class $\cK$ of similar smooth algebras. The definitions of 
majority and affine thin edges depend on this class.

% semilattice, affine and majority identities
We start with an observation that operations $f,g,h$ identified in
Corollary~\ref{cor:uniform} can be assumed to satisfy certain identities.

\begin{lemma}\label{lem:fgh-identities}
Operations $f,g,h$ identified in Corollary~\ref{cor:uniform} can be chosen 
such that
\begin{itemize}
\item[(1)]
$f(x,f(x,y))=f(x,y)$ for all $x,y\in\zA$ and all $\zA\in\cK$;
\item[(2)]
$g(x,g(x,y,y),g(x,y,y))=g(x,y,y)$ for all $x,y\in\zA$ and all $\zA\in\cK$;
\item[(3)]
$h(h(x,y,y),y,y)=h(x,y,y)$ for all $x,y\in\zA$ and all $\zA\in\cK$.
\end{itemize}
\end{lemma}

\begin{proof}
Items (1) and (2) follow from Lemma~\ref{lem:good-functions}.

(3) Let $h_b(x)=h(x,b,b)$ for $b\in\zA$. The goal is to find $h$ such that
$h_b(h_b(x))=h_b(x)$ for all $b$ and all $x$. Let $h'_0(x,y,z)=h(x,y,z)$ and\lb $h'_{i+1}(x,y,z)= h'_i(h(x,y,y),y,z)$ for $i\ge0$. Then $h'_i(x,b,b)=h_b^i(x)$. Clearly, $h_b^{|A|!}$ is idempotent  for every $b\in\zA$, and thus $h'_{|A|!}(h'_{|A|!}(x,y,y),y,y)=h'_{|A|!}(x,y,y)$. It remains to show that every function $h_i(x,y,z)$ 
is a replacement for $h$. That is, for any affine edge $ab$,
$$
h_{i+1}(a,b,b)\eqc\th h_{i+1}(b,b,a)\eqc\th a,
$$
where $\th\in\Con(\Sg{a,b})$ witnesses that $ab$ is an affine edge.
By induction we have
\begin{align*}
h_{i+1}(a,b,b) &= h_i(h(a,b,b),b,b)\eqc\th h_i(a,b,b)\eqc\th a,\\
h_{i+1}(b,b,a) &= h_i(h(b,b,b),b,a) = h_i(b,b,a)\eqc\th a.
\end{align*}
As is easily seen, the resulting operation $h$ acts on semilattice and majority edges as prescribed by Theorem~\ref{the:uniform}.
\end{proof}

% thin semilattice
%%%%%%%%%%%%%%%%%%%%%%%%%%%%%%%%%%%%
\subsection{Semilattice edges}\label{sec:thin-semilattice}

In this section we focus on semilattice edges of the graph
$\cG(\zA)$. Note first that if one fixes a term operation $f$ such that
$f$ is a semilattice operation on every thick semilattice edge of
$\cG(\zA)$, then it is possible to define an orientation of every semilattice edge.
A semilattice edge $ab$ is oriented from $a$
to $b$ if $f(a\fac\th,b\fac\th)=f(b\fac\th,a\fac\th)=b\fac\th$, where $\th$ is a 
congruence witnessing that $ab$ is a semilattice edge.
Clearly, this orientation strongly depends on
the choice of the term operation~$f$. 

We shall now improve the choice of operation $f$ and restrict the kind of
semilattice edges we will use later. A semilattice edge $ab$ such that
the equality relation witnesses that it is a semilattice edge will be called 
a \emph{thin semilattice
edge}. A binary operation $\ell(x,y)$ is said to 
satisfy the \emph{Semilattice Shift Condition} (SLS condition for short), if
\begin{align}
& \text{for any $a,b\in\zA$, either $a=\ell(a,b)$ or the pair}
\tag{SLS}\label{cond:sls}\\
& \text{$(a,\ell(a,b))$ is a thin semilattice edge.}
\nonumber
\end{align}

\begin{prop}\label{pro:good-operation}
Let $\zA$ be a smooth idempotent algebra. There is an SLS binary
term operation $f$ of $\zA$ such that $f$ is a semilattice operation on
every thick semilattice edge of $\cG(\zA)$, and $f$ is the first projection on every majority or affine thick edge.
\end{prop}

\begin{proof}
Let $f$ be a binary term operation such that $f$ is semilattice on every
semilattice edge and $f(x, f(x, y)) = f(x, y)$. Let $a, b\in \zA$ be such that
$f(a, b)\ne a$, and set $b_0 = f(a, b)$ and $b_{i+1} = f(a, f(b_i, a))$
for $i > 0$.

\medskip

{\sc Claim 1.}
For any $i$, $f(a, b_i) = b_i$.

\smallskip

Indeed, $f(a, b_0) = f(a, f(a, b)) = f(a, b) = b_0$, and for any $i > 0$
$$
f(a, b_i) = f(a, f(a, f(b_{i-1}, a))) = f(a, f(b_{i-1}, a)) = b_i.
$$

Let $\zB_i = \Sg{a, b_i}$. Then $\zB_0\supseteq \zB_1\supseteq\ldots$, 
and there is $k$ with $\zB_{k+1} = \zB_k$.

\medskip

{\sc Claim 2.}
$f(a, b_k) = f(b_k, a) = b_k$.

\smallskip

Since $b_k\in \zB_{k+1} = \Sg{a, b_{k+1}}$, there is a term operation
$t$ such that $b_k = t(a, b_{k+1})$. Let $s(x, y) =
t(x, f(x, f(y, x)))$. For this operation we have
\begin{align*}
s(a, b_k) &= t(a, f(a, f(b_k, a))) = t(a, b_{k+1}) = b_k\\
s(b_k, a) &= t(b_k, f(b_k, f(a, b_k))) = t(b_k, f(b_k, b_k)) = b_k.
\end{align*}

This means that $ab_k$ is a semilattice edge, and the congruence witnessing
it is the equality relation. Since $\zA$ is smooth, every one of its thick semilattice 
edge is a subalgebra, and by the choice of $f$, it is a 
semilattice operation on every such pair. Moreover, by Claim~1 
$f(a, b_k) = f(b_k, a) = b_k$.

Let $\ell$ be the maximal among the numbers chosen as before
Claim 2 for all pairs $a, b$ with $f(a, b)\ne a$. Let $f_0 = f$,
and $f_{i+1}(x, y) = f(x, f(f_i(x, y), x))$ for $i\ge 0$. Let also
$f' = f_\ell$.

\medskip

{\sc Claim 3.} For any $a, b\in A$, either $f'(a, b) = a$, or the
pair $ac$, where $c = f'(a, b)$ is a thin semilattice edge.

\smallskip

If $f(a, b) = a$ then it is straightforward that $f'(a, b) = a$.
Suppose $f(a, b)\ne a$. Note that in the
notation introduced before Claim~1 $b_i=f_i(a,b)$. Indeed, we 
have $b_0=f(a,b)=f_0(a,b)$, and then 
\[
f_{i+1}(a,b)=f(a,f(f_i(a,b),a))=f(a,f(b_i,a))=b_{i+1}.
\]
We prove by induction on $i$ that $f_i(a,c)=f_i(c,a)=c$. Since $\zB_\ell = \zB_{\ell+1}$, where the $\zB_i$ are constructed 
as before, by Claim 2 $f(a, c) = f(c, a) = c$. This gives the base case
of induction. Suppose $f_i(a, c) = f_i(c, a) = c$. Then
\begin{align*}
f_{i+1}(a, c) &= f(a, f(f_i(a, c), a) = f(a, f(c, a)) = c\\
f_{i+1}(c, a) &= f(c, f(f_i(c, a), c) = f(c, f(c, c)) = c.
\end{align*}
Claim 3 is proved.

\smallskip

To complete the proof it suffices to check that $f'$ is a
semilattice operation on every (thick) semilattice edge of $\cG(\zA)$.
However, this is straightforward from the construction of
$f'$. Also, as $f$ is the first projection on every thick edge of every
non-semilattice type, so is $f'$.
\end{proof}

It will be convenient for us to denote binary operation $f$ that satisfies the
conditions of Theorem~\ref{the:uniform}, 
Lemma~\ref{lem:fgh-identities}(1), and
Proposition~\ref{pro:good-operation} by $\cdot$, that is, to write
$x\cdot y$ or just $xy$ for $f(x,y)$, whenever it does not cause a confusion. 
The fact that $ab$ is a thin
semilattice edge we will also denote by $a\le b$. In other words,
$a\le b$ if and only if $a\cdot b=b\cdot a=b$.

Next we show that those semilattice edges $ab$, for which 
the equality relation does not witness that they are semilattice edges, 
can be thrown out of the graph $\cG(\zA)$ such
that the graph remains connected. Therefore, we can assume that every
semilattice edge is thin. Let $\zA\in\cK$ be an algebra, $a,b\in\zA$, $\zB=\Sg{a,b}$, 
and $\th$ a congruence of $\zB$. Pair $ab$ is said to be \emph{minimal} 
with respect to $\th$ if for any $b'\in b\fac\th$, $b\in\Sg{a,b'}$.

\begin{corollary}\label{cor:thick-thin}
Let $ab$ be a semilattice edge, $\th$ a congruence of
$\Sg{a,b}$ that witnesses this, and $c\in a\fac\th$. Then there is
$d\in b\fac\th$ such that $cd$ is a thin semilattice edge. Moreover, $cd$ 
is a thin semilattice edge for every $d\in b\fac\th$ such that $cd$ is minimal 
with respect to $d$.
\end{corollary}

\begin{proof}
By Proposition~\ref{pro:good-operation} $cb=c$ or $c\le cb$. Since
$d=cb\in b\fac\th$ the former option is impossible. Therefore $cd$ is a
thin semilattice edge. To prove the last claim let $d\in b\fac\th$ be such that 
$cd$ is a minimal pair with respect to $\th$. Let also $d'=c\cdot d$; we have 
$c\le d'$. Then there is a term operation $p$ such that $p(c,d')=d$. Let 
$f'(x,y)=p(x,x\cdot y)$ and $d''=f'(d,c)$. Then $f'(c,d)=d$ and $d''\in b\fac\th$.
Again, there is a term operation $r$ such that $r(c,d'')=d$. Then for 
$f''(x,y)=r(x,f'(y,x))$ we have $f''(c,d)=r(c,f'(d,c))=r(c,d'')=d$ and 
$f''(d,c)=r(d,f'(c,d))=r(d,d)=d$.
\end{proof}

%%%%%%%%%%%%%%%%%%%%%%%%%%%%%%%%%%%
\subsection{Thin majority edges}\label{sec:thin-majority}

Here we introduce thin majority edges in a way similar to thin 
semilattice edges, although in a weaker sense.

A ternary term operation $g'$ is said to satisfy the \emph{majority
condition} (with respect to $\cK$) if it satisfies the identity from 
Lemma~\ref{lem:fgh-identities}(2) and $g'$ is a majority operation on 
every thick majority edge of every algebra from $\cK$. By 
Corollary~\ref{cor:uniform} an operation satisfying the
majority condition exists.

A pair $ab$ is called a \emph{thin majority edge} (with respect to $\cK$) if 
\begin{itemize}
\item[(*)] 
for any term operation $g'$ satisfying the majority condition, 
the subalgebras $\Sg{a,g'(a,b,b)},\Sg{a,g'(b,a,b)},\Sg{a,g'(b,b,a)}$
contain $b$.
\end{itemize}
The operation $g$ from Corollary~\ref{cor:uniform} does not have to 
satisfy any specific conditions on a thin majority edge, except what follows 
from its definition. Also, thin majority edges are directed, since $a,b$ 
in the definition occur asymmetrically.

If in addition to the condition above $ab$ is also a majority edge, a 
congruence $\th$ witnesses that, and $ab$ is a minimal pair with respect to $\th$,
we say that $ab$ is a \emph{special majority edge}.
We now show that thin majority edges have term operations on them 
that are similar to majority operations. 

Next we prove a property, Lemma~\ref{lem:thin-majority-triple}, of thin 
majority edges that will be useful in the future. We start with an auxiliary lemma.

\begin{lemma}\label{lem:maj-condition}
Let $\zA$ be a smooth algebra and $g'$ a ternary term operation 
that is a majority operation on every thick majority edge. 
\begin{itemize}
\item[(1)]
For every binary term operation $t$, the operations $g'(t(x,g'(x,y,z)),y,z),\lb
g'(x,t(y,g'(x,y,z)),z), g'(x,y,t(z,g'(x,y,z)))$ are majority on 
every\lb thick majority edge.
\item[(2)]
If for some $a,b\in\zA$, $g'(a,b,b)=b$ [or $g'(b,a,b)=b$, or $g'(b,b,a)=b$], 
then there is a term operation $g''$ satisfying the majority condition and 
such that $g''(a,b,b)=b$ [or $g''(b,a,b)=b$, or $g''(b,b,a)=b$, respectively].
\item[(3)]
If for some $a,b\in\zA$, $g'(a,a,b)=b$,
then there is a term operation $g''$ satisfying the majority condition and 
such that $g''(a,a,b)=b$.
\end{itemize}
\end{lemma}

\begin{proof}
(1) Let $ab$ be a majority edge of $\zA$ and $\th$ a congruence
of $\Sg{a,b}$ that witnesses that. Then $\zB=\Sg{a,b}\fac\th$ is a
2-element algebra that has a majority operation but does not 
have a semilattice term operation. This means that $t$ is a 
projection on $\zB$. If $t(x,y)=x$ on $\zB$, then  
\[
g'(t(x,g'(x,y,z)),y,z)=g'(x,y,z)
\]
on $\zB$, as required. If $t(x,y)=y$ on $\zB$, then
\[
g''(x,y,z)=g'(t(x,g'(x,y,z)),y,z)=g'(g'(x,y,z),y,z).
\]
It is now straightforward to verify that $g''$ is majority on $\zB$.
A proof for the other two operations is quite similar.

(2) To show the existence of an operation $g''$ satisfying the equation 
$g''(x,g''(x,y,y), g''(x,y,y))=g''(x,y,y)$ we use the construction from the 
proof of Lemma~\ref{lem:good-functions}. We consider the
unary operation $g_x(y)=g'(x,y,y)$ and  
$g''(x,y,z)=g_x^{n-1}(g'(x,y,z))$. By Lemma~\ref{lem:good-functions} $g''$ 
is a majority operation on every thick majority edge. We show that for any $i\in[n-1]$
it holds that $g^\dagger(a,b,b)=b$ [or $g^\dagger(b,a,b)=b$,  or $g^\dagger(b,b,a)=b$],
where $g^\dagger(x,y,z)=g_x^i(g'(x,y,z))$. We proceed by induction on $i$.
Assuming $g_x^0(g'(x,y,z))=g'(x,y,z)$, the base case follows from the conditions 
of the lemma. Suppose $g^\dagger(x,y,z)=g_x^i(g'(x,y,z))$ and $g^\dagger(a,b,b)=b$.
We need to check that $g_a(g^\dagger(a,b,b))=b$ [$g_b(g^\dagger(b,a,b))=b$,$g_b(g^\dagger(b,b,a))=b$], which can be easily verified
\begin{align*}
g_a(g^\dagger(a,b,b)) &=g'(a,g^\dagger(a,b,b),g^\dagger(a,b,b))=g'(a,b,b)=b,\\
g_b(g^\dagger(b,a,b)) &=g'(b,g^\dagger(b,a,b),g^\dagger(b,a,b))=g'(b,b,b)=b,\\
g_b(g^\dagger(b,b,a)) &=g'(b,g^\dagger(b,b,a),g^\dagger(b,b,a))=g'(b,b,b)=b.
\end{align*}

(3) Here we use a slightly different construction. Define operations $g_i$ inductively by setting 
\begin{align*}
g_1(x,y,z) &= g'(x,y,z) \quad\text{and}\\
g_{i+1} &= g'(x,g_i(x,y,y),g_i(x,y,z)).
\end{align*}
If $g_i(a,a,b)=b$, then $g_{i+1}(a,a,b)=g'(a,g_i(a,a,a),g_i(a,a,b))=g'(a,a,b)=b$. Also, if $g_i$ is a majority operation on some thick majority edge, it is straightforward to verify that so is $g_{i+1}$ on that edge. Therefore by the assumptions on $g'$, for every $i$ the operation $g_i$ is a majority operation on every thick majority edge of any $\zB\in\cK$, and $g_i(a,a,b)=b$. We now find $n$ such that $g_n$ satisfies the identity from Lemma~\ref{lem:good-functions}(2). Let $g_x(y)=g'(x,y,y)$, $n_x$ the idempotent power of $g_x$, and $n$ the least common multiple of $n_x$ for all $x\in\zB\in\cK$. Then we have 
\[
g_n(x,y,y)=\underbrace{g_x\circ g_x\circ\dots\circ g_x(y)}_{n\text{ times}}=g_x^n(y),
\]
and so 
\[
g_n(x,g_n(x,y,y),g_n(x,y,y))=g_x^n(g_x^n(y))=g_x^n(y)=g_n(x,y,y),
\]
as required.
\end{proof}

\begin{lemma}\label{lem:thin-majority-triple}
Let $\zA_1,\zA_2,\zA_3\in\cK$, and let $a_1b_1$, $a_2b_2$, and $a_3b_3$ 
be thin majority edges in 
$\zA_1,\zA_2,\zA_3$. Then there is an operation $g'$ satisfying the majority 
condition such that $g'(a_1,b_1,b_1)=b_1$, 
$g'(b_2,a_2,b_2)=b_2$, $g'(b_3,b_3,a_3)=b_3$.
In particular, for any thin majority edge $ab$ there is an operation 
$g'$ satisfying the majority condition such that 
$g'(a,b,b)=g'(b,a,b)=g'(b,b,a)=b$.
\end{lemma}

\begin{proof}
Let $g$ be an operation satisfying the conditions of Corollary~\ref{cor:uniform};
in particular, it satisfies the majority condition. 
Let $\rel$ be the subalgebra of $\zA_1\tm\zA_2\tm\zA_3$ generated by
$\ba_1=(a_1,b_2,b_3),\ba_2=(b_1,a_2,b_3),\ba_3=(b_1,b_2,a_3)$. Let 
$b^1_1=g(a_1,b_1,b_1)$. By the definition of thin
majority edges there is a term operation $t$ such that $b_1=t(a_1,b^1_1)$.
Consider $g'_1(x,y,z)=g(t(x,g(x,y,z)),y,z)$,
by Lemma~\ref{lem:maj-condition}(1) $g'_1$ is a majority operation on every 
thick majority edge and 
\[
g'_1(a_1,b_1,b_1)=g(t(a_1,g(a_1,b_1,b_1)),b_1,b_1)
=g(b_1,b_1,b_1)=b_1.
\]
By Lemma~\ref{lem:maj-condition}(2) $g'_1$ can be converted into $g'$ 
satisfying the majority condition and such that $g'(a_1,b_1,b_1)=b_1$. 
Also let $g'(b_2,a_2,b_2)=b^2_2$ and $g'(b_3,b_3,a_3)=b^2_3$,
so that $g'(\ba_1,\ba_2,\ba_3)=(b_1,b^2_2,b^2_3)$.

Again, as $g'$ satisfies the majority condition, by the definition of thin
majority edges there is a term operation $s$ such that $b_2=s(a_2,b^2_2)$.
Consider $g''_1(x,y,z)=g'(x,s(y,g'(x,y,z)),z)$,
by Lemma~\ref{lem:maj-condition}(1) $g''_1$ is a majority operation on every 
thick majority edge and 
\[
g''_1(b_2,a_2,b_2)=g'(b_2,s(a_2,g'(b_2,a_2,b_2)),b_2)
=g'(b_2,b_2,b_2)=b_2.
\]
Also 
\[
g''_1(a_1,b_1,b_1)=g'(a_1,s(b_1,g'(a_1,b_1,b_1)),b_1)
=g'(a_1,s(b_1,b_1),b_1)=b_1.
\]
Operation $g''_1$ can be transformed into $g''$ satisfying the majority 
condition and such that $g''(a_1,b_1,b_1)=b_1, g''(b_2,a_2,b_2)=b_2$.
Let $g''(b_3,b_3,a_3)=b^3_3$,
so that $g''(\ba_1,\ba_2,\ba_3)=(b_1,b_2,b^3_3)$. 

Now we use the same construction once more. As $g''$ satisfies 
the majority condition, there is a binary term 
operation $q$ such that $q(a_3,b^3_3)=b_3$. Consider operation 
$g'''_1(x,y,z)=g''(x,y,q(z,g''(x,y,z)))$. 
As before, we obtain
\[
g'''_1(b_3,b_3,a_3)=g''(b_3,b_3,q(a_3,g''(b_3,b_3,a_3)))
=g''(b_3,b_3,b_3)=b_3.
\]
Also 
\[
g'''_1(a_1,b_1,b_1)=g''(a_1,b_1,q(b_1,g''(a_1,b_1,b_1)))
=g''(a_1,b_1,q(b_1,b_1))=b_1,
\]
and 
\[
g'''_1(b_2,a_2,b_2)=g''(b_2,a_2,q(b_2,g''(b_2,a_2,b_2)))
=g''(b_2,a_2,q(b_2,b_2))=b_2.
\]
Thus, $g'''_1(\ba_1,\ba_2,\ba_3)=(b_1,b_2,b_3)$. We now apply 
Lemma~\ref{lem:maj-condition}(2) to obtain $g'''$ that satisfies the 
conditions of the lemma.
\end{proof}

Next we show that every majority edge has a thin majority edge 
associated with~it.

% thin majority, combined majority
\begin{lemma}\label{lem:thin-majority}
Let $\zA\in\cK$ be a smooth idempotent algebra, $ab$ a majority 
edge in $\zA$,  
and $\th$ a congruence of $\Sg{a,b}$ witnessing that. Then for 
any $c\in a\fac\th$ and $d\in b\fac\th$ such that $cd$ is a minimal 
pair with respect to the restriction $\th'$ of $\th$ on $\Sg{c,d}$, 
the pair $cd$ is a special thin majority edge.
\end{lemma}

\begin{proof}
Since $g'(c,d,d)\eqc\th g'(d,c,d)\eqc\th g'(d,d,c)\eqc\th d$ for 
any $g'$ satisfying the majority condition and $cd$ is a minimal 
pair with respect to $\th'$, we have $d\in\Sg{c,g'(c,d,d)},
\Sg{c,g'(d,c,d)},\Sg{c,g'(d,d,c)}$. 
\end{proof}

Since for any $c\in a\fac\th, d\in b\fac\th$ such that $\Sg{c,d}$ is minimal 
among the subalgebras of the form $\Sg{c,d'}$, $d'\in b\fac\th$, the pair 
$cd$ is minimal with respect to $\th$, we obtain the following

\begin{corollary}\label{cor:thin-majority}
For any majority edge $ab$, where $\th$ is a witnessing congruence, 
and any $c\in a\fac\th$ there is $d\in b\fac\th$ such that $cd$ is a special 
thin majority edge.
\end{corollary}

A natural question is, of course, whether it is also possible to find a 
`symmetric' thin majority edge in any thick majority edge. In other words, 
is it true that for any majority edge $ab$ (witnessed by congruence 
$\th$ of $\Sg{a,b}$) there are $c\in a\fac\th, d\in b\fac\th$ such that 
$cd$ is a majority edge witnessed by the equality relation? 
Unfortunately, the following example derived from one suggested by 
M.Kozik \cite{Kozik17:no-symmetry} shows that this is not true in 
general.

\begin{example}\label{exa:no-majority-symmetry}\rm
Let $A=\{0,1,2,3\}$, and let $\th$ be the equivalence relation on $A$ 
with blocks $\{0,2\}$ and $\{1,3\}$. Define two ternary operations 
$\maj$ and $\min$ on $A$ as follows: $\maj$ is majority and 
$\min$ is the first projection on $A\fac\th$. On each of the 
$\th$-blocks $\{0,2\},\{1,3\}$, the operation $\maj$ is the third 
projection, and $\min$ is the minority operation. Finally, for any 
$a,b,c\in A$ such that $(b,c)\in\th$, but $(a,b)\not\in\th$ we set 
$\maj(a,b,c)=\maj(b,a,c)=c$, $\maj(b,c,a)=\maj(b,c,a-1)$, 
$\min(a,b,c)=\min(a+2,b-1,c-1)$, $\min(b,a,c)=\min(b,a-1,c)$, $\min(b,c,a)=\min(b,c,a-1)$. 
Here $+$ and $-$ are modulo 4. Also, note that in the given conditions 
$\maj(b,c,a-1),\min(a+2,b-1,c-1),\min(b,a-1,c),\min(b,c,a-1)$ are
defined by the action of $\min,\maj$ on $\th$-blocks.
Let $\zA=(A,\maj,\min)$. As is easily seen, $\zA$ omits type \one, and 
any pair $ab$, $a\in\{0,2\},b\in\{1,3\}$, is a majority edge, as 
witnessed by the congruence $\th$. It can be 
verified by straightforward computation (use Universal Algebra 
Calculator \cite{Freese:ucalc} for that) that for no such pair there 
is a term operation of $\zA$ that is majority on $\{a,b\}$. 
\end{example}

%%%%%%%%%%%%%%%%%%%%%%%%%%%%%%%%%
\subsection{Thin affine edges}\label{sec:thin-affine}

In this section we introduce thin affine edges in a similar fashion as
thin majority and semilattice edges.

We say that a term operation $h'$ of a smooth algebra $\zA\in\cK$ 
satisfies the \emph{minority condition} (with respect to $\cK$) if 
it satisfies the identity from Lemma~\ref{lem:fgh-identities}(3), and
for any $\zB\in\cK$ and every affine edge $ab$ of $\zB$ witnessed 
by a congruence $\th$ of $\Sgg\zB{a,b}$, the operation $h'$ is a Mal'tsev 
operation on $\Sgg\zB{a,b}\fac\th$. By Corollary~\ref{cor:uniform} and 
Lemma~\ref{lem:fgh-identities}(3)
an operation satisfying the minority condition exists.

A pair $ab$ is called a \emph{thin affine edge} (with respect to $\cK$) 
if 
\begin{itemize}
\item[(**)] 
$h(b,a,a)=b$, where $h$ is the operation identified in Corollary~\ref{cor:uniform}, and $b\in\Sg{a,h'(a,a,b)}$ for every term operation $h'$ satisfying the minority condition.
\end{itemize}
The operation $h$ from
Corollary~\ref{cor:uniform} does not have to satisfy any specific 
conditions on thin minority edges, except what follows from its 
definition. Also, thin affine edges are directed, since $a,b$ in 
the definition occur asymmetrically. 

\begin{lemma}\label{lem:thin-affine}
Let $\zA_1,\zA_2\in\cK$, and
let $ab$ and $cd$ be thin affine edges in $\zA_1,\zA_2$. Then 
there is an operation $h'$ and such that $h'(b,a,a)=b$ and 
$h'(c,c,d)=d$. In particular, for any thin affine edge $ab$ there 
is an operation $h'$ such that $h'(b,a,a)=h'(a,a,b)=b$.
\end{lemma}

\begin{proof}
Let $\rel$ be the subalgebra of $\zA_1\tm\zA_2$ generated by 
$(b,c),(a,c),(a,d)$. By the definition of thin affine edges,
$$
\cl b{d'}=h\left(\cl bc,\cl ac,\cl ad\right)\in\rel,
$$
where $h$ is an operation satisfying the minority condition that exists by 
Corollary~\ref{cor:uniform} and Lemma~\ref{lem:fgh-identities}(3). 
Then as $d\in\Sgg{\zA_2}{c,h(c,c,d)}$, there is a term operation 
$r(x,y)$ such that $d=r(c,d')$. Therefore 
\[
\cl bd=r\left(\cl bc,\cl b{d'}\right)\in\rel.
\] 
The result follows.
\end{proof}

% thin affine, combined Mal'tsev operation
\begin{lemma}\label{lem:thin-affine2}
Let $\zA\in\cK$ be an algebra, $ab$ an affine edge in it, and $\th$ the 
congruence of $\Sg{a,b}$ witnessing that. Then there exists  
$b'\in b\fac\th$ such that $ab'$ is a minimal pair and $h(b',a,a)=b'$. Moreover, 
for any such $b'$, the pair $ab'$ is a thin affine edge.
\end{lemma}

\begin{proof}
Suppose $a,b$ satisfy the conditions of the lemma. Choose $b''$
such that $ab''$ is a minimal pair.
Let $b'=h(b'',a,a)\in\Sg{a,b''}$, since $h$ satisfies the conditions of 
Lemma~\ref{lem:fgh-identities}(3) 
$h(b',a,a)=b'$. Note that by the choice of $b''$ the pair $ab'$ is 
minimal with respect to the restriction $\th'$ of
$\th$ on $\Sg{a,b'}$. Since $h'(a,a,b')\in b\fac\th$ for any $h'$ 
satisfying the minority condition, this means
that $b'\in\Sg{a,h'(a,a,b')}$. Condition (**) follows.
\end{proof}

\begin{corollary}\label{cor:thin-affine}
For any affine edge $ab$, where $\th$ is a witnessing 
congruence, there is $b'\in b\fac\th$ such that $ab'$ is a thin 
affine edge. 
\end{corollary}

%%%%%%%%%%%%%%%%%%%%%%%%%%%%%%%%%%
\subsection{Some useful terms}

First, we make a useful observation concerning the presence of thick and
thin edges of certain types.

\begin{prop}\label{pro:thin-thick-colors}
Let $\cK$ be a finite class of smooth algebras without edges of the 
unary type. Then there exists an algebra from $\cK$ containing a 
thin semilattice (majority, affine) edge if and only if there is an
algebra in $\cK$ containing a thick edge of the same type.
\end{prop}

\begin{proof}
Corollaries~\ref{cor:thick-thin},~\ref{cor:thin-majority},
and~\ref{cor:thin-affine} imply that if an algebra $\zA$ contains a
thick edge of a certain type, it also contains a thin edge of the same
type. Conversely, every thin semilattice edge is also a thick one by 
definition. Suppose that $\cK$ does not contain algebras with 
majority (affine) edges. Then any operation satisfying the conditions of 
Lemma~\ref{lem:fgh-identities} also satisfies the majority 
(minority) condition, including the first projection $p(x,y,z)$.
Let $a,b\in\zA\in\cK$, $a\ne b$. Then $p(a,b,b)=a$ ($p(a,a,b)=a$),
and so $\Sgg\zA{a,p(a,b,b)}=\{a\}$ ($\Sgg\zA{a,p(a,a,b)}=\{a\}$).
Therefore, $ab$ cannot be a thin edge of the majority (affine) type.
\end{proof}

The next two lemmas show that smooth algebras always have a range of
term operations that behave in a predictable way on thin or thick edges. 
These properties will be very helpful later when proving various statements 
about the structure of smooth algebras.

\begin{lemma}\label{lem:thin-combinations}
Let $\zA_1,\zA_2\in\cK$, and let $ab$ and $cd$ be thin edges in $\zA_1,\zA_2$.
If they have different types, then there is a binary term operation $p$ such that
$p(b,a)=b$, $p(c,d)=d$. 
\end{lemma}

\begin{proof}
Let $\rel$ be the subalgebra of $\zA_1\tm\zA_2$ generated by 
$(b,c),(a,d)$. Let also $g,h$ be the operations from 
Corollary~\ref{cor:uniform}.

If $ab$ is majority and $c\le d$, then by the definition of thin 
majority edges there is a binary term operation $r$ such that 
$b=r(a,g(a,b,b))$. Then
$$
\cl bd=r\left(\cl ad,g\left(\cl ad,\cl bc,\cl bc\right)\right)\in\rel,
$$
as $g(x,y,z)=xyz$ on semilattice edges. Therefore $p(x,y)=r(y,g(y,x,x))$ 
satisfies the conditions.

If $ab$ is affine and $c\le d$, then by the definition of thin 
affine edges,
$$
\cl bd=h\left(\cl bc,\cl ad,\cl ad\right)\in\rel,
$$
as $h(x,y,z)=xyz$ on semilattice edges. Therefore $p(x,y)=h(x,y,y)$ 
satisfies the conditions.

If $ab$ is affine and $cd$ is majority, then set 
$$
\cl {b'}{d'}=h\left(\cl ad,\cl ad,\cl bc\right)\in\rel.
$$
By the definition of thin affine edges there is a binary term operation $r$ 
such that $b=r(a,h(a,a,b))$. Consider operation
\[
g'(x,y,z)=g(r(x,h(x,y,z)),r(y,h(y,x,z)),z).
\]
It satisfies the following condition.
\begin{align*}
g'(a,a,b) & =g(r(a,h(a,a,b)),r(a,h(a,a,b)),b)\\
& =g(r(a,b'),r(a,b'),b)=g(b,b,b)=b.
\end{align*}
Also, on any thick majority edge $\{c'\fac\th,d'\fac\th\}$ of an algebra from $\cK$, 
where $\th$ witnesses that $c'd'$ is a majority edge, $h(x,y,z)=x$, therefore
\begin{align*}
g'(x,y,z) &= g(r(x,h(x,y,z)),r(y,h(y,x,z)),z)\\
&= g(r(x,x),r(y,y),z)=g(x,y,z)
\end{align*}
on $\Sgg{\zA_2}{c',d'}\fac\th$. 

By Lemma~\ref{lem:maj-condition}(3) there is an the operation $g''$ that satisfies the majority condition and such that $g''(a,a,b)=b$. As $cd$ is a thin majority edge, there is a binary term operation $s$ such that $s(c,g''(d,d,c))=d$.
Thus,
\[
\cl bd=s\left(\cl bc,\cl b{g''(d,d,c)}\right)=
s\left(\cl bc,g''\left(\cl ad,\cl ad,\cl bc\right)\right)\in\rel.
\]
The result follows.
\end{proof}

\begin{lemma}\label{lem:op-s-on-affine}
(1) Let $ab$ be a thin majority edge of an algebra $\zA\in\cK$. There is a term
operation $t_{ab}$\label{not:1-ab} such that $t_{ab}(a,b)=b$ and $t_{ab}
(c,d)\eqc\eta c$ for any affine edge $cd$ of any $\zA'\in\cK$, where the
type of $cd$ is witnessed by congruence $\eta$.\\[1mm]
(2) The operation $h'(x,y,z)=h(z,y,x)$ satisfies the following conditions: $h'(a,a,b)=b$ for any thin affine edge $ab$ of any $\zA'\in\cK$, and $h'(d,c,c)\eqc\eta d$ for any affine edge $cd$ of any $\zA'\in\cK$, where the type of $cd$ is witnessed by congruence $\eta$.
Moreover, $h'(x,c',d')$ is a permutation of $\Sg{c,d}\fac\eta$
for any $c',d'\in\Sg{c,d}$.
\end{lemma}

\begin{proof}
Let $g,h$ be operations satisfying the conditions of Corollary~\ref{cor:uniform}.

(1) Let $c_1d_1\zd c_\ell d_\ell$ be a list of all affine edges of algebras in
$\cK$, $c_i,d_i\in\zA_i$ and $\th_i$ the corresponding congruences. 
Let $b'=g(a,b,b)$. By the definition of thin majority edges 
$b\in\Sgg\zA{a,b'}$ and there is a binary term operation $r$ such that
$b=r(a,b')$. By Corollary~\ref{cor:uniform} $g$ is the first projection on 
$\Sg{c_i,d_i}\fac{\th_i}$. Let $t_{ab}(x,y)=r(x,g(x,y,y))$. Then
\begin{align*}
t_{ab}(a,b) &= r(a,g(a,b,b))=b,\\
t_{ab}(c_i,d_i) &= r(c_i,g(c_i,d_i,d_i))\eqc{\th_i}c_i, \quad\text{for $i\in[\ell]$}.
\end{align*}
This means that $t_{ab}$ satisfies the required conditions.

(2) The first result follows from the definition of thin affine edges and the fact that $h$ is a Mal'tsev operation on $\Sg{c,d}\fac{\eta}$ for every affine edge $cd$ and congruence $\eta$ witnessing that. 

Let $c',d'\in\Sg{c,d}$. Since $\zB=\Sg{c,d}\fac{\eta}$ is a module, in particular, it is an Abelian algebra and $h'(x,c^*,c^*)=x$ for all $c^*\in\zB$, the second result follows.
\end{proof}

%%%%%%%%%%%%%%%%%%%%%%%%%%%%%%%%%%
%%%%%%%%%%%%%%%%%%%%%%%%%%%%%%%%%%
\section*{Acknowledgements}

The author would like to thank Marcin Kozik, Libor Barto, Keith Kearnes, and the anonymous referees for a number of helpful hints and suggestions.

%%%%%%%%%%%%%%%%%%%%%%%%%%%%%%%%%%
%%%%%%%%%%%%%%%%%%%%%%%%%%%%%%%%%%
\section*{Declarations}

\subsection*{Data availability}
Data sharing not applicable to this article as datasets were neither generated nor analyzed.

\subsection*{Compliance with ethical standards}
The author is a member of the Editorial Board of Algebra Universalis. Apart from this the author declares that he has no conflict of interest.

\bibliographystyle{spmpsci}
%% \bibliography{one}

\end{document}